\documentclass[journal,twoside,web]{ieeecolor}
\usepackage{generic}
\usepackage{cite}
\usepackage{textcomp}

\usepackage{graphicx}
\usepackage{xspace}
\usepackage{float}
\usepackage{booktabs}
\usepackage{hyperref}
\hypersetup{hidelinks=true}

\usepackage[dvipsnames]{xcolor}
\usepackage{tikz}

\usepackage{amsmath,amssymb,amsfonts}
\usepackage{mathrsfs}
\usepackage{url}
\usepackage{nicematrix}
\usepackage{empheq}
\usepackage{caption}
\usepackage{subcaption}
\usepackage[capitalize]{cleveref}

\let\labelindent\relax
\usepackage{enumitem}

\newtheorem{lemma}{Lemma}
\newtheorem{definition}{Definition}
\newtheorem{theorem}{Theorem}
\newtheorem{remark}{Remark}
\newtheorem{proposition}{Proposition}
\newtheorem{conjecture}{Conjecture}
\newtheorem{corollary}{Corollary}
\newtheorem{assumption}{Assumption}

\usepackage{tikz}
\usepackage{xspace}

\DeclareRobustCommand{\legendline}[1]{\hspace{-2pt}
\tikz[#1,line width=1.5pt,baseline=-0.5ex]{\draw (0,0) -- (.35,0);}
\hspace{-2pt}}

\definecolor{mblue}{rgb}{0,0.4470,0.7410}
\definecolor{morange}{rgb}{0.8500,0.3250,0.0980}
\definecolor{myellow}{rgb}{0.9290,0.6940,0.1250}
\definecolor{mpurple}{rgb}{0.4940,0.1840,0.5560}
\definecolor{mgreen}{rgb}{0.4660,0.6740,0.1880}
\definecolor{mcyan}{rgb}{0.3010,0.7450,0.9330}
\definecolor{mred}{rgb}{0.6350,0.0780,0.1840}
\definecolor{mgreenblue}{rgb}{0.0,1.0,0.5}
\definecolor{parulablue}{rgb}{0.2431,0.1490,0.6588}
\definecolor{parulalblue}{RGB}{39,151,235}
\definecolor{parulagreen}{RGB}{129,204,89}
\definecolor{parulayellow}{RGB}{249,251,21}
\definecolor{wintergreen}{cmyk}{0.61,0,0.74,0}

\newcommand{\e}[1]{\ensuremath{\cdot 10^{#1}}\xspace}

\DeclareFontFamily{OT1}{pzc}{}
\DeclareFontShape{OT1}{pzc}{m}{it}{ <-> s*[1.1] pzcmi7t }{}
\DeclareMathAlphabet{\mathpzc}{OT1}{pzc}{m}{it}

\newcommand{\mc}[1]{\mathcal{#1}}

\newcommand{\mr}[1]{\mathrm{#1}}
\newcommand{\mf}[1]{\mathfrak{#1}}
\newcommand{\mb}[1]{\mathbb{#1}}
\newcommand{\ms}[1]{\mathscr{#1}}

\newcommand{\mpz}[1]{\mathpzc{#1}}

\newcommand{\Partial}[2]{\frac{\partial #1}{\partial #2}}

\newcommand{\C}[1]{\ensuremath{\mc{C}_{#1}}\xspace}   

\newcommand{\lp}[1]{\ensuremath{\mathcal{L}_{#1}}\xspace}

\newcommand{\ltwo}{\lp{2}}

\newcommand{\linf}{\lp{\infty}}

\newcommand{\lpe}[1]{\ensuremath{\mathcal{L}_{#1\mr{e}}}\xspace}

\newcommand{\ltwoe}{\lpe{2}} 

\newcommand{\lsp}[1]{\ensuremath{\mathcal{L}_{\mr{s}#1}}\xspace}    

\newcommand{\lstwo}{\lsp{2}}
\newcommand{\lsinf}{\lsp{\infty}}

\newcommand{\posClassSym}{\mc{Q}}
\newcommand{\posClass}[1]{\posClassSym_{#1}}

\newcommand{\reals}{{\mb{R}}}

\newcommand{\nnreals}{{\reals_0^+}}

\newcommand{\sym}{\ensuremath{\mb{S}}\xspace}

\DeclareMathOperator{\col}{col}
\DeclareMathOperator{\diag}{diag}

\DeclareMathOperator{\proj}{\pi}

\newcommand{\norm}[1]{\left\lVert#1\right\rVert}

\newcommand{\st}{x}		          
\newcommand{\stSet}{\mc{X}}      
\newcommand{\stSize}{{n_\mr{x}}}   

\newcommand{\stIc}{\st_0}

\newcommand{\ip}{u}		
\newcommand{\ipSet}{\mc{U}}
\newcommand{\ipSize}{{n_\mr{u}}}

\newcommand{\op}{y}		
\newcommand{\opSet}{\mc{Y}}
\newcommand{\opSize}{{n_\mr{y}}}

\newcommand{\gd}{w}   
\newcommand{\gdSet}{\mc{W}}
\newcommand{\gdSize}{{n_\mr{w}}}

\newcommand{\gp}{z}   
\newcommand{\gpSet}{\mc{Z}}
\newcommand{\gpSize}{{n_\mr{z}}}

\newcommand{\stEq}{\st_*}

\newcommand{\gdEq}{\gd_*}
\newcommand{\gpEq}{\gp_*}

\newcommand{\eqSet}{\ms{E}}
\newcommand{\stSetEq}{\ms{X}}
\newcommand{\gdSetEq}{\ms{W}}
\newcommand{\gpSetEq}{\ms{Z}}

\newcommand{\eqMap}{\kappa}

\newcommand{\stk}{\st_\mr{k}}	  
\newcommand{\ipk}{\ip_\mr{k}}
\newcommand{\opk}{\op_\mr{k}}

\newcommand{\stkMap}{f_\mr{k}}
\newcommand{\opkMap}{h_\mr{k}}

\newcommand{\stkSet}{\mc{X}_\mr{k}}
\newcommand{\stkSize}{{n_\mr{x_k}}}

\newcommand{\ipkSize}{{n_\mr{u_k}}}

\newcommand{\opkSize}{{n_\mr{y_k}}}

\newcommand{\stclSet}{\mc{X}_\mr{cl}}

\newcommand{\stclIc}{x_\mr{cl,0}} 

\newcommand{\rf}{r}  
\newcommand{\dist}{d} 

\newcommand{\weight}{W}   
\newcommand{\weightint}{M}
\newcommand{\filter}{F}
\newcommand{\stf}[1]{\st_{\mr{F},#1}}
\newcommand{\ipf}[1]{\ip_{\mr{F},#1}}
\newcommand{\opf}[1]{\op_{\mr{F},#1}}
\newcommand{\stfSize}[1]{{n_{\mr{F},#1}}}

\newcommand{\perf}{\gamma}  

\newcommand{\icfunShift}{\zeta_\mr{s}}

\newcommand{\genplant}{P}

\newcommand{\controller}{K}

\newcommand{\plant}{G}

\newcommand{\velogenplant}{P_\mr{v}}
\newcommand{\velogenplantLPV}{P_\mr{vpv}}
\newcommand{\velocontroller}{K_\mr{v}}

\newcommand{\stp}{\st}
\newcommand{\ipp}{\ip}
\newcommand{\opp}{\op}

\newcommand{\Bp}{\B}  
\newcommand{\Binv}[1]{\B_{#1}}

\newcommand{\stpIc}{\stIc}

\newcommand{\stpSet}{\stSet}
\newcommand{\stpSetLPV}{\stSetLPV}
\newcommand{\ippSet}{\ipSet}
\newcommand{\ippSetLPV}{\ipSetLPV}
\newcommand{\oppSet}{\opSet}

\newcommand{\stpSize}{\stSize}
\newcommand{\ippSize}{\ipSize}
\newcommand{\oppSize}{\opSize}

\newcommand{\stMap}{f}
\newcommand{\opMap}{h}

\newcommand{\opgpMap}{h_\mr{\gp}}
\newcommand{\opopMap}{h_\mr{\op}}

\newcommand{\B}{\mf{B}^\mr{NL}}  
\newcommand{\Bw}{\B_\mr{w}} 

\newcommand{\Bc}{\mf{B}^\mr{C}} 
\newcommand{\Bcw}{\Bc_\mr{\gd}} 

\newcommand{\var}{\lambda}

\newcommand{\intA}{\bar A}  

\newcommand{\sttran}{\phi_\mr{\st}}

\newcommand{\stdot}{\st_\mr{v}}

\newcommand{\ltiB}{B}
\newcommand{\ltiC}{C}

\newcommand{\ltiBw}{B_\mr{\gd}} 
\newcommand{\ltiBu}{B_\mr{\ip}}

\newcommand{\ltiCy}{C_\mr{\op}}

\newcommand{\ltiDzw}{D_\mr{\gp\gd}}

\newcommand{\ltiDyw}{D_\mr{\op\gd}}

\newcommand{\lpvA}{A}
\newcommand{\lpvB}{B}
\newcommand{\lpvC}{C}
\newcommand{\lpvD}{D}

\newcommand{\lpvBu}{B_\mr{\ip}}
\newcommand{\lpvCz}{C_\mr{\gp}}
\newcommand{\lpvCy}{C_\mr{\op}}

\newcommand{\lpvDzu}{D_\mr{\gp\ip}}

\newcommand{\lpvDyu}{D_\mr{\op\ip}}

\newcommand{\lpvAk}{A_\mr{k}} 
\newcommand{\lpvBk}{B_\mr{k}}
\newcommand{\lpvCk}{C_\mr{k}}
\newcommand{\lpvDk}{D_\mr{k}}

\newcommand{\sch}{p}		
\newcommand{\schSet}{\mc{P}}
\newcommand{\schSize}{{n_\mr{p}}}

\newcommand{\schMap}{\eta}

\newcommand{\Blpv}[1][\sch]{\mf{B}^\mr{VPV}_{#1}}  
\newcommand{\Blpvfull}{\mf{B}^\mr{VPV}}  

\newcommand{\stSetLPV}{\mpz{X}}
\newcommand{\gdSetLPV}{\mpz{W}}
\newcommand{\ipSetLPV}{\mpz{U}}

\newcommand{\storquad}{M}  

\newcommand{\lpvClass}{\mathfrak{A}}

\newcommand{\dt}{\xi}    

\newcommand{\tSet}{\mc{T}} 

\newcommand{\ic}[2]{\mc{F}_\mr{l}(#1,#2)} 

\newcommand{\lyapfun}{V}
\newcommand{\supfun}{s}
\newcommand{\storfun}{\mc{V}}

\newcommand{\supQ}{Q}
\newcommand{\supR}{R}
\newcommand{\supS}{S}

\newcommand{\qsr}{(\supQ,\supS,\supR)}

\newcommand{\qsrMat}{\begin{bmatrix} \supQ  & \supS\\ \star & \supR \end{bmatrix}}

\newcommand{\storfunIncr}{\storfun_\mr{i}}

\newcommand{\otherTraj}[1]{\expandafter\tilde #1}
\newcommand{\sto}{\otherTraj{\st}}

\newcommand{\gdo}{\otherTraj{\gd}}
\newcommand{\gpo}{\otherTraj{\gp}}

\newcommand{\parTraj}[1]{\expandafter\bar #1}

\newcommand{\dotB}{\partial}

\newcommand{\dif}{\partial}

\newcommand{\Bv}{\mathfrak{B}^\mr{VF}}
\newcommand{\Bvw}{\Bv_\mr{\gd}}
\newcommand{\Bvset}[1]{\Bv_\mr{#1}}

\newcommand{\lyapfunShift}{\lyapfun_\mr{s}}
\newcommand{\supfunShift}{\supfun_\mr{s}}
\newcommand{\storfunShift}{\storfun_\mr{s}}

\newcommand{\lyapfunVelo}{\lyapfun_\mr{v}}
\newcommand{\supfunVelo}{\supfun_\mr{v}}
\newcommand{\storfunVelo}{\storfun_\mr{v}}

\newcommand{\velA}{A_\mr{v}}
\newcommand{\velB}{B_\mr{v}}
\newcommand{\velC}{C_\mr{v}}
\newcommand{\velD}{D_\mr{v}}

\newcommand{\velBu}{B_\mr{v,\ip}}
\newcommand{\velCz}{C_\mr{v,\gp}}
\newcommand{\velCy}{C_\mr{v,\op}}

\newcommand{\velDzu}{D_\mr{v,\gp\ip}}

\newcommand{\velDyu}{D_\mr{v,\op\ip}}

\newcommand{\gddot}{\gd_\mr{v}}
\newcommand{\ipdot}{\ip_\mr{v}}
\newcommand{\opdot}{\op_\mr{v}}
\newcommand{\gpdot}{\gp_\mr{v}}

\newcommand{\ippdot}{\ipdot}
\newcommand{\oppdot}{\opdot}
\newcommand{\stpdot}{\st_\mr{v}}

\newcommand{\ipkdot}{\ip_\mr{v,k}}
\newcommand{\opkdot}{\op_\mr{v,k}}
\newcommand{\stkdot}{\st_\mr{v,k}}

\newcommand{\exoA}{A_\mr{\gd}}
\newcommand{\exoBvr}{\mf{W}}

\newcommand{\BcXU}{\Bc_{\stSetLPV\ippSetLPV}}
\newcommand{\BvXU}{\Bv_{\stSetLPV\ippSetLPV}}

\newcommand{\stkus}{\breve{x}_\mr{k}}
\newcommand{\dotstkus}{\dot{\breve{x}}_\mr{k}}

\newcommand{\stkA}{\tilde{x}_\mr{k}}
\newcommand{\dotstkA}{\dot{\tilde{x}}_\mr{k}}
\newcommand{\stkB}{\hat{x}_\mr{k}}
\newcommand{\dotstkB}{\dot{\hat{x}}_\mr{k}}

\newcommand{\lpvAkus}{\breve{A}_\mr{k}}
\newcommand{\lpvBkus}{\breve{B}_\mr{k}}
\newcommand{\lpvCkus}{\breve{C}_\mr{k}}
\newcommand{\lpvDkus}{\breve{D}_\mr{k}}

\newcommand{\veloform}[1]{#1_\mr{v}}

\def\extendedversion{1} 

\newcommand{\extver}[2]{%
  \ifx\extendedversion\undefined%
	#2\xspace
  \else%
   #1\xspace
  \fi%
}

\newcommand{\customurl}[2]{\href{#1}{\path{#2}}}
\newcommand{\proofsection}[1]{\subsection{Proof of \cref{#1}}}

\makeatletter
\g@addto@macro\normalsize{%
  \setlength\abovedisplayskip{.3em}
  \setlength\belowdisplayskip{.3em}
  \setlength\abovedisplayshortskip{.3em}%
  \setlength\belowdisplayshortskip{.3em}%
}

\let\olddot\dot 
\gdef\dot{\expandafter\olddot}

\let\oldddot\ddot 
\gdef\ddot{\expandafter\oldddot}

\crefformat{equation}{(#2#1#3)}   
\crefrangeformat{equation}{(#3#1#4)--(#5#2#6)}
\crefmultiformat{equation}{(#2#1#3)}{ and~(#2#1#3)}{, (#2#1#3)}{, and~(#2#1#3)}
\crefrangemultiformat{equation}{(#3#1#4)--(#5#2#6)}{ and~(#3#1#4)--(#5#2#6)}{, #3#1#4--(#5#2#6)}{, and~#3#1#4--(#5#2#6)}

\crefformat{assumption}{Assumption~#2#1#3}
\crefrangeformat{assumption}{Assumptions~#3#1#4--#5#2#6}
\crefmultiformat{assumption}{Assumptions~#2#1#3}{ and~#2#1#3}{, #2#1#3}{, and~#2#1#3}
\crefrangemultiformat{assumption}{Assumptions~#3#1#4--#5#2#6}{ and~#3#1#4--#5#2#6}{, #3#1#4--#5#2#6}{, and~#3#1#4--#5#2#6}

\crefformat{contribution}{\normalfont Contribution~#2#1#3}
\crefrangeformat{contribution}{\normalfont Contributions~#3#1#4--#5#2#6}
\crefmultiformat{contribution}{\normalfont Contributions~#2#1#3}{ and~#2#1#3}{, #2#1#3}{, and~#2#1#3}
\crefrangemultiformat{contribution}{\normalfont Contributions~#3#1#4--#5#2#6}{ and~#3#1#4--#5#2#6}{, #3#1#4--#5#2#6}{, and~#3#1#4--#5#2#6}

\Crefname{conjecture}{Conjecture}{Conjectures}

\def\BibTeX{{\rm B\kern-.05em{\sc i\kern-.025em b}\kern-.08em
    T\kern-.1667em\lower.7ex\hbox{E}\kern-.125emX}}
    
\begin{document}
\title{Equilibrium-Independent Control of Continuous-Time Nonlinear Systems via the LPV Framework\extver{\\\vspace{.3em}\normalsize{(Extended Version)}\vspace{-.5em}}{}}
\author{Patrick J. W. Koelewijn, Siep Weiland, and Roland T\'oth, \IEEEmembership{Senior Member, IEEE}\extver{\vspace{-1em}}{}
\thanks{This work has received funding from the European Research Council (ERC) under the European Union Horizon 2020 research and innovation programme (grant agreement No 714663) and was also supported by the European Union within the framework of the National Laboratory for Autonomous Systems (RRF-2.3.1-21-2022-00002).}
\thanks{P.J.W. Koelewijn, S. Weiland, and R. T\'oth are with the Control Systems Group, Department of Electrical Engineering, Eindhoven University of Technology, Eindhoven, The Netherlands (e-mails: p.j.w.koelewijn@tue.nl, s.weiland@tue.nl, r.toth@tue.nl). R. T\'oth is also with the Systems and Control  Laboratory, Institute for Computer Science and Control, Budapest, Hungary.}}

\maketitle

\begin{abstract}
In this paper, we consider the analysis and control of continuous-time nonlinear systems to ensure universal shifted stability and performance, i.e., stability and performance w.r.t. each forced equilibrium point of the system. This ``equilibrium-free'' concept is especially beneficial for control problems that require the tracking of setpoints and rejection of persistent disturbances, such as input loads.  In this paper, we show how the velocity form, i.e., the time-differentiated dynamics of the system, plays a crucial role in characterizing these properties and how the analysis of it can be solved by the application of Linear Parameter-Varying (LPV) methods in a computationally efficient manner. Furthermore, by leveraging the properties of the velocity form and the LPV framework, a novel controller synthesis method is presented which ensures closed-loop universal shifted stability and performance. The proposed controller design is verified \extver{in a simulation study and also}{}experimentally on a real system. Additionally, we compare the proposed method to a standard LPV control design, demonstrating the improved stability and performance guarantees of the new approach.
\end{abstract}


\section{Introduction}\label{4_sec:introduction}
\IEEEPARstart{T}{he} analysis and control of nonlinear systems becomes increasingly important as we push for progressively higher performance requirements and as systems become more and more complex. However, in the industry, \emph{Linear Time-Invariant} (LTI) methods are still widely used, as for LTI systems there is an extensive, systematic, and computationally efficient framework for analysis and control design that allows to ensure and shape stability and performance of the closed-loop system. While there exist a multitude of analysis and controller synthesis methods for nonlinear systems, so far, a systematic framework for analysis and controller synthesis for nonlinear systems has not been introduced like is available for LTI systems. While approaches such as the \emph{Linear Parameter-Varying} (LPV) framework have aimed to achieve this, they are unable to do so, as the stability and performance guarantees of the current state-of-the-art LPV methods are dependent on the choice of equilibrium point \cite{Koelewijn2020}.

Standard stability and dissipativity are sufficient if we only want to analyze these properties w.r.t. a single (forced) equilibrium point of the system. However, they become troublesome to use if we want to ensure them w.r.t. all (forced) equilibrium points of the system. This is especially relevant in cases when one wants to track constant references and/or reject constant (unknown) disturbances. Consequently, stability and performance concepts which are not dependent on equilibrium points of the system are highly important to arrive at a systematic framework for nonlinear analysis and controller synthesis. In the literature,  notions such as so-called \emph{shifted} stability/dissipativity \cite{Van-der-Schaft2017} and \emph{equilibrium independent} stability/dissipativity \cite{Jayawardhana2006,Hines2011,Simpson-Porco2019} have been introduced, whereby stability/dissipativity w.r.t. a particular (non-zero) (forced) equilibrium point is ensured, or w.r.t. all forced equilibrium points of the system, respectively. In literature, equilibrium independent dissipativity has also been referred to as \emph{constant incremental dissipativity} \cite{Jayawardhana2006}. In order to not confuse this notion of stability/dissipativity with other notions, we will refer to it as \emph{Universal Shifted} (US) stability/dissipativity.

In literature, the analysis of non-zero equilibrium points of continuous-time nonlinear systems has also been investigated through its time-differentiated dynamics. In gain-scheduling and the LPV framework, this has it roots in the so-called velocity-based scheduling techniques \cite{Leith1998a,Leith1999,Kaminer1995,Toth2010}, which generally have improved performance compared to normal linearization based methods. However, these results are based on the argument that locally around an equilibrium point, the time-differentiated dynamics coincide with the linearization of the nonlinear system at the equilibrium point. Consequently, these results are only able to provide local guarantees in a neighborhood around the equilibrium points, which severely hampers their viability. From a nonlinear perspective, the time-differentiated dynamics also connect to the so-called Krasovskii method for the construction of a Lyapunov function to show stability \cite{Khalil2002}. This has also been explored more recently in connection to non-zero equilibrium point stability and/or performance properties, in  \cite{Kosaraju2019,Kawano2021,Schweidel2022}. The work in \cite{Kawano2021} uses time-differentiated properties to ensure US stability and US passivity of the system. However, while the overall results are quite interesting, the output map of the nonlinear state-space representation of the underlying system is required to be a particular form to prove their implications and results have only been obtained for passivity. The work in \cite{Schweidel2022} focuses on the analysis of the network interconnections of systems, where it is shown that (unique) equilibrium points of interconnections of velocity dissipative systems are stable. However, no connection to performance is made in this work or how these results could be used for controller design. There exist controller synthesis methods that ensure even stronger notions of stability and performance, such as incremental stability and performance \cite{Koelewijn2020b}, which can also be used to ensure US stability and performance. However, the controller realization in \cite{Koelewijn2020b} is complicated and might very demanding in specific applications, as it requires the target trajectory to be known in advance. Therefore, one could wonder what is really needed to ensure \emph{Universal Shifted Stability} (USS) and \emph{Universal Shifted Performance} (USP). Consequently, as contributions of this paper, we will
\begin{enumerate}[label={\bfseries C\arabic*:},ref={C\arabic*}]
\item Show how analysis of the time-differentiated dynamics can be used to imply both USS and USP of continuous-time nonlinear systems. This generalizes the existing USS results which use the Krasovskii condition. 
	\label[contribution]{con:velo}
\item Show how the analysis of the time-differentiated dynamics can be performed through the LPV framework to systematically and computationally efficiently analyze USS and USP.\label[contribution]{con:lpv}
\item Present a novel controller realization scheme that enables synthesis and implementation of nonlinear controllers in practice such that the closed-loop has guaranteed USS and USP. \label[contribution]{con:syn}
\end{enumerate}
Compared to the incremental controller design in \cite{Koelewijn2020b}, the novel realization scheme we propose in this paper results in a controller that has a simpler structure and does not require explicit knowledge of the equilibrium points in order to ensure USS and USP. This is a particular advantage in practice, with the price being less strong stability and performance implications.

In \cref{4_sec:shifted}, we introduce USS, USP, and \emph{Universal Shifted Dissipativity} (USD). In \cref{4_sec:veloanalysis}, we discuss velocity-based analysis and how it implies USS and USP. \cref{4_sec:velousinglpv} shows velocity-based analysis through the LPV framework. In \cref{4_sec:contrsynthesis}, we develop a US controller synthesis method. The performance and properties of the method are demonstrated in \cref{4_sec:examples}\extver{, both through a simulation study and}{\xspace through}experiments on a unbalanced disk system. Finally, in \cref{4_sec:conclusion}, conclusions on the established synthesis and analysis toolchain are drawn.

\subsubsection*{Notation}
$\reals$ is the set of real numbers, while $\nnreals$ is the set of non-negative reals. We denote by $\sym^n$ the set of real symmetric matrices of size $n\times n$. 
Denote the projection operation by $\proj$, s.t. for $\mc{D}=\mc{A}\times\mc{B}$, $\proj_\mr{a}\mc{D}\subseteq\mc{A}$ and s.t. for any $a\in\proj_\mr{a}\mc{D}$, there exists a $b$ s.t. $(a,b)\in\mc{D}$ and for any $b$ there exists \emph{no} $a\notin\proj_\mr{a}\mc{D}$ s.t. $(a,b)\in\mc{D}$. 
For a signal $\gd:\nnreals\to\reals^n$, denote by $\gd\equiv\gd_*$ that $\gd(t)=\gd_*\in\reals^n$ for all $t\in\nnreals$. 
$\C{n}$ is the class $n$-times continuously differentiable functions. A function $V:\reals^n\to\reals$ belongs to the class $\posClass{\stEq}$ if it is positive definite and decrescent w.r.t. $\stEq\in\reals^n$ (see \cite[Definition 3.3]{Scherer2015}).
$\lpe{p}$ denotes the (extended) space of integrable functions $f:[0,\,T]\to\reals^n$ with $T\in\nnreals$ and with finite $p$-norm $\norm{f}_{p,T} = ({\int_0^T \norm{f(t)}^p dt})^{\frac{1}{p}}$, where $\norm{\star}$ is the Euclidean (vector) norm.
We use $(\star)$ to denote a symmetric term {in a quadratic expression, e.g., $(\star)^\top  Q(a-b) = (a-b)^\top Q(a-b)$ for $Q\in\sym^n$ and $a,b\in\reals^{n}$}. 
The notation $A\succ 0$ ($A\succeq 0$) indicates that $A\in\sym^n$ is positive (semi-) definite, while $A\prec 0$ ($A\preceq 0$) {denotes a} negative (semi-)definite $A\in\sym^n$. Furthermore, $\col(x_1, \dots ,x_n)$ denotes the column vector $[x_1^\top \cdots x_n^\top]^\top$.




\section{Universal Shifted Stability and Performance}\label{4_sec:shifted}
\subsection{Nonlinear system}
In this section, we will formally introduce the concepts of USS, USP, and USD. We consider nonlinear dynamical systems given by
\begin{subequations}\label{4_eq:nonlinsys}
\begin{align}
	\dot \st(t) &= \stMap(\st(t),\gd(t));\label{4_eq:nonlinsyssteq}\\
	\gp(t) &= \opMap(\st(t),\gd(t));\label{4_eq:nonlinsysopeq}
\end{align}
\end{subequations}
where $t\in\nnreals$ is time, $\st(t)\in\stSet\subseteq \reals^\stSize$ is the state with initial condition $\st(0)=\stIc\in\reals^\stSize$, $\gd(t)\in\gdSet\subseteq\reals^\gdSize$ is the input of the system, and $\gp(t)\in\gpSet\subseteq\reals^\gpSize$ is the output of the system. Moreover, $\stMap:\stSet\times\gdSet\to\reals^\stSize$ and $\opMap:\stSet\times\gdSet\to\gpSet$ are functions with $\stMap,\opMap\in\C{1}$. Define \begin{multline}
	\B = \lbrace (\st,\gd,\gp)\in(\stSet\times\gdSet\times\gpSet)^{\nnreals}\mid \st\in\C{1}\text{ and }\\ (\st,\gd,\gp)\text{ satisfy \cref{4_eq:nonlinsys}} \rbrace,
\end{multline}
as the set of solutions of \cref{4_eq:nonlinsys} which are assumed to be forward complete and unique. We denote the behavior of \cref{4_eq:nonlinsys} for a specific input trajectory $\bar\gd\in\gdSet^\nnreals$ by
\begin{equation}
	\Bw(\bar\gd) = \{ (\st,\gd,\gp)\in\B \mid \gd=\bar\gd\in\gdSet^\nnreals \}.
\end{equation}
We also define the state transition map $\sttran:\nnreals\times\nnreals\times\stSet\times\gdSet^\tSet\to\stSet$, such that 
$\st(t) = \sttran(t,0,\stIc,\gd)$.

For the system given by \cref{4_eq:nonlinsys}, the set of equilibrium points is defined as 
\begin{multline}\label{4_eq:eqsetdef}
	\eqSet=\lbrace (\stEq,\gdEq,\gpEq)\in\stSet\times \gdSet\times\gpSet\mid 0 = \stMap(\stEq,\gdEq),\\ \gpEq = \opMap(\stEq,\gdEq) \rbrace.
\end{multline}
Define $\stSetEq = \proj_\mr{\stEq}\eqSet$, $\gdSetEq = \proj_\mr{\gdEq}\eqSet$, and $\gpSetEq = \proj_\mr{\gpEq}\eqSet$. Throughout this paper, we make the following assumption:
\begin{assumption}\label{4_assum:uniqueEq}
	For the nonlinear system given by \cref{4_eq:nonlinsys}, we assume that there exists a bijective map $\eqMap:\gdSetEq\to \stSetEq$ such that $\stEq = \eqMap(\gdEq)$, for all $(\stEq,\gdEq)\in \proj_\mr{\stEq,\gdEq}\eqSet$.
\end{assumption}
This assumption means that for each $\gdEq\in\gdSetEq$ there is a unique corresponding $\stEq\in\stSetEq$, and vice versa. This assumption is taken for convenience, to not overcomplicate the discussion.

\subsection{Universal shifted stability}
As mentioned in \cref{4_sec:introduction}, USS is stability w.r.t. all forced equilibrium points of the system. Defined more formally as:
\begin{definition}[Universal shifted stability]\label{4_def:shiftedstab}
	The nonlinear system given by \cref{4_eq:nonlinsys} is \emph{Universally Shifted Stable} (USS), if for each $\epsilon>0$ and $(\stEq,\gdEq)\in \proj_\mr{\stEq,\gdEq}\eqSet$, there exists a $\delta(\epsilon)>0$ s.t. $\norm{\st(0)-\stEq} < \delta(\epsilon)$ implies that $\norm{\st(t)-\stEq}< \epsilon$ for all $\st\in\Bw(\gd\equiv\gdEq)$ and $t\in\nnreals$. The system is \emph{Universally Shifted Asymptotically Stable} (USAS) if it is USS and for $\gd\equiv\gdEq$ we have that $\norm{\st(0)-\stEq} < \delta(\epsilon)$ implies that $\lim_{t\to\infty}\norm{\sttran(t,0,\st(0),\gd)-\stEq}=0$.
\end{definition}
Note that this definition is nothing more than Lyapunov stability, see e.g. \cite{Khalil2002,Van-der-Schaft2017}, for each equilibrium point of the system. 
We can extend the standard Lyapunov condition to analyze US(A)S:
\begin{theorem}[Universal shifted Lyapunov stability]\label{4_thm:shiftlyapstab}
	The nonlinear system given by \cref{4_eq:nonlinsys} is USS, if there exists a function $\lyapfunShift:\stSet\times\gdSetEq\to\nnreals$ with $\lyapfunShift(\cdot,\gdEq)\in\C{1}$ and $\lyapfunShift(\cdot,\gdEq)\in\posClass{\stEq}$ for every $(\stEq,\gdEq)\in\proj_\mr{\stEq,\gdEq}\eqSet$, such that
	\begin{equation}\label{4_eq:shiftedstability}
		\Partial{}{t}\lyapfunShift(\st(t),\gdEq)\leq 0,
	\end{equation}
	for all $t\in\nnreals$, $\st\in\proj_\mr{\st}\Bw(\gd\equiv \gdEq)$, and for every $(\stEq,\gdEq)\in\proj_\mr{\stEq,\gdEq}\eqSet$. If \cref{4_eq:shiftedstability} holds, but with strict inequality except when $\st(t)=\stEq$, then the system is USAS.
\end{theorem}
\begin{proof}
\extver{See \cref{4_pf:shiftlyapstab}.}{See \cite{Koelewijn2023a}. The proof is also given in the extended version of this paper \cite{Koelewijn2023}.}
\end{proof}

\subsection{Universal shifted dissipativity}\label{4_sec:shifteddissip}

\emph{Universal Shifted Dissipativity} (USD) extends the classical dissipativity concept by Willems \cite{Willems1972} w.r.t. each (forced) equilibrium point of the system instead of a single point of neutral storage. This allows for analyzing both US(A)S and \emph{Universal Shifted Performance} (USP) of the system. In literature, USD has also been referred to as equilibrium independent dissipativity \cite{Hines2011,Simpson-Porco2019} or constant incremental dissipativity \cite{Jayawardhana2006}.
%
More concretely, we take the following definition for USD, adopted from \cite{Simpson-Porco2019}:
\begin{definition}[Universal shifted dissipativity]\label{4_def:shifteddissip}
	The nonlinear system given by \cref{4_eq:nonlinsys} is \emph{Universally Shifted Dissipative} (USD) w.r.t. the supply function $\supfunShift:\gdSet\times\gdSetEq\times \gpSet\times \gpSetEq\to\reals$, if there exists a storage function $\storfunShift:\stSet\times\gdSetEq\to\nnreals$ with $\storfunShift(\cdot,\gdEq)\in\C{0}$ and $\storfunShift(\cdot,\gdEq)\in\posClass{\stEq}$ for every $(\stEq,\gdEq)\in\proj_\mr{\stEq,\gdEq}\eqSet$, such that
	\begin{equation}\label{4_eq:shifteddissip}
	\storfunShift(\st(t_1),\gdEq)-\storfunShift(\st(t_0),\gdEq)\leq \!\int_{t_0}^{t_1}\!\! \supfunShift(\gd(t),\gdEq,\gp(t),\gpEq)\,dt,
\end{equation}
for all $t_0,t_1\in\nnreals$ with $t_1\geq t_0$ and $(\st,\gd,\gp)\in\B$.
\end{definition}

Classical dissipativity is connected to many well-known performance metrics, such as \ltwo-gain and passivity, by considering quadratic, so-called $\qsr$, supply functions. Similarly, also for USD, we consider quadratic supply functions $\supfunShift$ of the form
\begin{equation}\label{4_eq:shiftsupply}
	\supfunShift(\gd,\gdEq,\gp,\gpEq) = \begin{bmatrix}
		\gd-\gdEq\\ \gp-\gpEq
	\end{bmatrix}^\top \qsrMat\begin{bmatrix}
		\gd-\gdEq\\ \gp-\gpEq
	\end{bmatrix},
\end{equation}
where $\supQ\in\sym^\gdSize$, $\supR\in\sym^\gpSize$, and $\supS \in\reals^{\gdSize\times\gpSize}$. A system given by \cref{4_eq:nonlinsys} is $\qsr$-USD if it is USD with respect to a supply function of the form \cref{4_eq:shiftsupply}. Moreover, it can easily be seen that if a system is $\qsr$-USD and $(0,0,0)\in\eqSet$, it is also classically $\qsr$-dissipative for the same tuple $\qsr$. Therefore, USD is a stronger notion than classical dissipativity. 



\subsection{Induced universal shifted performance}\label{4_sec:shiftdissipperf}
Before connecting $\qsr$-USD to quadratic performance notions, we will first give a definition for the US extension of the \lp{p}-\lp{q}-gain:
\begin{definition}[Universal shifted \lp{p}-\lp{q}-gain]\label{4_def:shiftlplqgain}
	A nonlinear system given by \cref{4_eq:nonlinsys} is said to have a finite US \lp{p}-\lp{q}-gain, if there is a finite $\perf\geq 0$ and function $\icfunShift :\stSet\times\stSetEq \to\reals$ such that for every $(\stEq,\gdEq)\in\proj_\mr{\stEq,\gdEq}\eqSet$ it holds that
	\begin{equation}\label{4_eq:shiftlplqgain}
		\norm{\gp-\gpEq}_{q,T} \leq \perf \norm{\gd-\gdEq}_{p,T} + \icfunShift(\stIc,\stEq),
	\end{equation}
	for all $T\geq 0$ and $(\st,\gd,\gp)\in\B$ with\,\footnote{Note that this implies $(\gd-\gdEq)\in\lpe{p}$.} 
	$\gd\in\lpe{p}$. The induced US \lp{p}-\lp{q}-gain of \cref{4_eq:nonlinsys}, denoted as \lsp{p}-\lsp{q}-gain, is the infimum of $\perf$ such that \cref{4_eq:shiftlplqgain} still holds. If $p=q$, we will refer to this as the (induced) US \lp{p}-gain, denoted as \lsp{p}-gain.
\end{definition}
Using this definition, we directly have US extensions of the well-known \ltwo-gain, \linf-gain, and \ltwo-\linf-gain through the \lstwo-gain, \lsinf-gain, and \lstwo-\lsinf-gain, respectively. We can then connect $\qsr$-USD to \lstwo-gain performance through the following lemma:
\begin{lemma}[\lstwo-gain based on USD]\label{4_lem:ls2gaindissip}
If the nonlinear system given by \cref{4_eq:nonlinsys} is $\qsr$-USD with $\qsr = (\perf^2 I,0,-I)$, then the system has an \lstwo-gain bounded by $\perf$.
\end{lemma}
\begin{proof}
\extver{See \cref{4_pf:ls2gaindissip}.}{See \cite{Koelewijn2023a}. The proof is also given in the extended version of this paper \cite{Koelewijn2023}.}
\end{proof}

For US passivity, we have adopted the following definition inspired by the shifted passivity definition in \cite{Van-der-Schaft2017}:
\begin{definition}\label{4_def:shiftpass}
	A nonlinear system given by \cref{4_eq:nonlinsys} is said to be US passive, if it is $\qsr$-USD with $\qsr=(0,I,0)$.
\end{definition}
Through this definition, we directly link US passivity to USD. Note that if $(0,0,0)\in\eqSet$, then, the US performance notions also imply their standard counterparts in terms of the \lp{p}-\lp{q}-gain and passivity. 

\subsection{Induced universal shifted stability}\label{4_sec:shiftdissipstab}
Classical dissipativity of a nonlinear system implies stability (at the origin) of the system if the supply function satisfies a negativity condition, see e.g. \cite{Van-der-Schaft2017,Brogliato2020}. We will show how a similar condition on the US supply function $\supfunShift$ can be formulated to link USD and US(A)S.
\begin{theorem}[USS from USD]\label{4_thm:shiftdissipstab}
If the nonlinear system given by \cref{4_eq:nonlinsys} is USD w.r.t. a supply function $\supfunShift$  under a storage function $\storfunShift$ with $\storfunShift(\cdot,\gdEq)\in\C{1}$ for all $\gdEq\in\gdSetEq$ and $\supfunShift$ satisfies for every $(\stEq,\gdEq,\gpEq)\in\eqSet$ that
	\begin{equation}\label{4_eq:supfunshiftcond}
		\supfunShift(\gdEq,\gdEq,\gp,\gpEq)\leq 0,
	\end{equation}	
	for all $\gp\in\gpSet$, then, the nonlinear system is USS. If the supply function satisfies \cref{4_eq:supfunshiftcond}, but with strict inequality for all $\gp\neq \gpEq$, and the system is observable (see \cite[Definition 3.27]{Nijmeijer2016}), then the nonlinear system is USAS.
\end{theorem}
\begin{proof}
\extver{See \cref{4_pf:shiftdissipstab}.}{See \cite{Koelewijn2023a}. The proof is also given in the extended version of this paper \cite{Koelewijn2023}.}
\end{proof}
Note that the condition in \cref{4_thm:shiftdissipstab} is satisfied for US $\qsr$ supply functions for which $\supR\preceq 0$ (with $\supR\prec 0$ implying USAS, as is the case for the \lstwo-gain).

\section{Velocity-based Analysis}\label{4_sec:veloanalysis}
\subsection{The velocity form}\label{4_sec:veloform}
In this paper, we are interested in analyzing and ensuring US(A)S and USP of nonlinear systems. In \cref{4_sec:shifted}, we have shown how USD allows us to simultaneously analyze both US(A)S and USP of a system. While \cref{4_def:shifteddissip} gives us a condition to analyze USD, these conditions require to hold for all state and input trajectories/values and for every equilibrium point $(\stEq, \gdEq, \gpEq) \in \eqSet$. Hence, checking USD directly through these conditions can be difficult, if not an infeasible task. Next, as one of the main results of this paper, we will show how analysis of the time-differentiated dynamics of the system will allow us to simplify US(A)S and USP analysis of systems.

Let us first define the following restriction of the solution set of \cref{4_eq:nonlinsys} by\,\footnote{As solutions are defined on $\nnreals$, we assume they are also continuously differentiable at $t=0$.}
\begin{equation}\label{4_eq:Bc}
	\Bc = \{(\st,\gd,\gp)\in\B\mid \st\in\C{2},\, \gd,\gp\in\C{1}\},
\end{equation}
i.e., the solutions in $\B$ that are differentiable. We also define $\Bcw(\gd)= \Bw(\gd)\cap \Bc$ for a $\gd \in \gdSet^\nnreals$. Furthermore, define the operator $\dotB$ for these sets such that
\begin{multline}\label{4_eq:velosolution}
	\dotB \Bc = \big\{(\dot \st,\dot \gd,\dot \gp)\in (\reals^\stSize\times \reals^\gdSize\times \reals^\gpSize)^{\nnreals} \mid \dot \st(t)=\tfrac{d}{dt}\st(t),\\ \dot{\gd}(t) = \tfrac{d}{dt}\gd(t), \dot \gp(t) = \tfrac{d}{dt}\gp(t), \forall t\in\nnreals, (\st,\gd,\gp) \in \Bc\big\}.
\end{multline}
We refer to the time-differentiated dynamics of the nonlinear system given by \cref{4_eq:nonlinsys} as the velocity form of the system.
\begin{definition}[Velocity form]\label{4_def:veloform}
For a nonlinear system given by \cref{4_eq:nonlinsys}, the velocity form is
\begin{subequations}\label{4_eq:veloform}
	\begin{align}
		\ddot{\st}(t) &= \velA(\st(t),\gd(t))\dot{\st}(t)+\velB(\st(t),\gd(t))\dot{\gd}(t);\label{4_eq:velostate}\\
	\dot{\gp}(t) &= \velC(\st(t),\gd(t))\dot{\st}(t)+\velD(\st(t),\gd(t))\dot{\gd}(t);\label{4_eq:velooutput}
\end{align}
\end{subequations}
where $\velA=\Partial{\stMap}{\st}$, $\velB=\Partial{\stMap}{\gd}$, $\velC=\Partial{\opMap}{\st}$, $\velD=\Partial{\opMap}{\gd}$, and $(\st,\gd,\gp)\in\Bc$.
\end{definition}

The solution set of \cref{4_eq:veloform} is given by $\Bv=\dotB\Bc$, and we define $\Bvw(\gd) = \dotB \Bcw(\gd)$ for a $\gd \in \gdSet^\nnreals$.

We will refer to the original system given by \cref{4_eq:nonlinsys} as the primal form of the system. As aforementioned, the analysis of the time-differentiated dynamics has been investigated in connection to gain scheduling \cite{Leith1998a,Leith1999}, the construction of (local) LPV embeddings \cite{Toth2010}, and more recently also in connection with (universal) shifted stability/dissipativity \cite{Kosaraju2019,Kawano2021,Schweidel2022}. However, the existing works on this topic are limited as they only provide local guarantees, assume severe restrictions to the output dynamics, and/or focus solely on stability.

For the velocity form \cref{4_eq:veloform}, we can also define a notion of dissipativity, which connects to US(A)S and USP:
\begin{definition}[Velocity dissipativity]\label{4_def:velodissip}
	The nonlinear system given by \cref{4_eq:nonlinsys} is \emph{Velocity Dissipative} (VD) w.r.t. the supply function $\supfunVelo:\reals^\gdSize\times\reals^\gpSize\to\reals$, if there exists a storage function $\storfunVelo:\reals^\stSize\to\nnreals$ with $\storfunVelo\in\C{1}$ and $\storfunVelo\in\posClass{0}$, such that, for all $t_0,t_1\in\nnreals$ with $t_1\geq t_0$,
\begin{equation}\label{4_eq:velodissip}
	\storfunVelo(\dot{\st}(t_1))-\storfunVelo(\dot{\st}(t_0))\leq \int_{t_0}^{t_1} \supfunVelo(\dot{\gd}(t),\dot{\gp}(t))\,dt,
\end{equation}
for all $(\dot{\st},\dot{\gd},\dot{\gp})\in\Bv$.
\end{definition}
Note that in this sense, VD can be seen as `classical dissipativity' of the velocity form \cref{4_eq:veloform} of the system. In literature, similar notions have also been introduced, such as \emph{Krasovskii passivity} \cite{Kosaraju2019,Kawano2021} and \emph{delta dissipativity} \cite{Schweidel2022}, which are also connected to non-zero equilibrium point properties of the primal form of the system.

Similarly to the USD, also for VD, we focus on quadratic supply functions of the form
\begin{equation}\label{4_eq:velosupply}
	\supfunVelo(\dot{\gd},\dot{\gp}) = \begin{bmatrix}
		\dot{\gd}\\ \dot{\gp}
	\end{bmatrix}^\top\qsrMat\begin{bmatrix}
		\dot{\gd}\\ \dot{\gp}
	\end{bmatrix},
\end{equation}
where again $\supQ\in\sym^\gdSize$, $\supS \in\reals^{\gdSize\times\gpSize}$, and $\supR\in\sym^\gpSize$. If a system is VD w.r.t. a supply function of the form \cref{4_eq:velosupply}, we will refer to it being $\qsr$-VD.


If we consider a $\qsr$ supply function \cref{4_eq:velosupply} and quadratic storage function:
\begin{equation}\label{4_eq:velostorquad}
	\storfunVelo(\dot\st) = \dot\st^\top \storquad \dot\st,
\end{equation}
where $\storquad\in\sym^\stSize$ with $\storquad \succ 0$, the following sufficient condition for $\qsr$-VD can be derived:
\begin{theorem}[$\qsr$-VD condition]\label{4_thm:veloqsrMI}
	The system given by \cref{4_eq:nonlinsys} is $\qsr$-VD, if there exists an $\storquad\in\sym^\stSize$ with $\storquad \succ 0$, such that, for all $(\st,\gd)\in\stSet\times\gdSet$,
\begin{multline}\label{4_eq:veloMI}
(\star)^\top  \begin{bmatrix} 0 &\storquad \\\star &0 \end{bmatrix}  
\begin{bmatrix} 
I & 0 \\ 
\velA(\st,\gd) & \velB(\st,\gd)
\end{bmatrix} 
-\\(\star)^\top \qsrMat \begin{bmatrix} 
0 & I \\ 
\velC(\st,\gd) & \velD(\st,\gd)
\end{bmatrix}  \preceq 0.
\end{multline}
\vspace{0em}
\end{theorem}
\begin{proof}
\extver{See \cref{4_pf:veloqsrMI}}{See \cite[Example 1]{Schweidel2022} or the extended version of this paper \cite{Koelewijn2023}}.
\end{proof}

Note, what is compelling about the condition given in \cref{4_thm:veloqsrMI} is that for a fixed $(\st,\gd)\in\stSet\times\gdSet$, \cref{4_eq:veloMI} becomes a \emph{Linear Matrix Inequality} (LMI). Later, in \cref{4_sec:velousinglpv}, we will see how we can use tools from the LPV framework to reduce this infinite-dimensional set of LMIs to a finite-dimensional set, which can be computationally efficiently verified, giving an efficient tool to analyze $\qsr$-VD of a system.

\subsection{Induced universal shifted stability}\label{4_sec:veloshiftstab}
Next, we will first present results on how the velocity form and VD connect to US(A)S of the system. Let us first introduce the set $\Bvset{\gdSetEq}=\cup_{\gdEq\in\gdSetEq} \Bvw( \gd\equiv \gdEq)$, i.e., the behavior of the velocity form for which the input is $\gd(t)=\gdEq\in\gdSetEq$, and hence $\dot\gd(t)=0$, for all $t\in\nnreals$.

\begin{theorem}[Implied USS]\label{4_thm:velotoshiftstab}
	The nonlinear system given by \cref{4_eq:nonlinsys}, with solutions in $\Bc$, is USS, if there exists a function $\lyapfunVelo:\reals^{\stSize}\to\nnreals$ with $\lyapfunVelo\in\C{1}$ and $\storfunVelo\in\posClass{0}$, such that
	\begin{equation}\label{4_eq:velostability}
		\Partial{}{t}\lyapfunVelo(\dot{\st}(t))\leq 0,
	\end{equation}
	for all $t\in\nnreals$ and $\dot\st\in\proj_\mr{\dot{\st}}\Bvset{\gdSetEq}$. If \cref{4_eq:velostability} holds, but with strict inequality except when $\dot{\st}(t)=0$, then the system is USAS.
\end{theorem}
\begin{proof}
See \cref{4_pf:velotoshiftstab}.
\end{proof}
The proof of \cref{4_thm:velotoshiftstab} relies on a construction of a US Lyapunov function based on $\lyapfunVelo$ and is based on the so-called Krasovskii method for Lyapunov stability \cite{Khalil2002}. A similar method is also used for the results in \cite{Kawano2021,Schweidel2022}.

Note that the condition in \cref{4_thm:velotoshiftstab} can be interpreted as the velocity form being stable (w.r.t. $\dot\st=0$), similar to how VD can be seen classically dissipativity of the velocity form. This is also intuitive, since $\dot\st(t)=0$ corresponds to an equilibrium point of the system. For $t\to \infty$, we have that $\dot\st(t)\to 0$ in the asymptotic stability case, meaning the state converges to an equilibrium point, i.e., $\st(t)\to\stEq\in\stSetEq$. 

We can also connect VD and US(A)S through \cref{4_thm:velotoshiftstab} and by restricting the velocity supply function $\supfunVelo$:
\begin{theorem}[USS from VD]\label{4_lem:velostab}
	Assume the nonlinear system given by \cref{4_eq:nonlinsys} is VD under a storage function $\storfunVelo\in\C{1}$ w.r.t. a supply function $\supfunVelo$ that satisfies
	\begin{equation}\label{4_eq:supplystability}
		\supfunVelo(0,\gpdot)\leq 0,
	\end{equation}	
	for all $\gpdot\in\reals^\gpSize$, then, the nonlinear system is USS. If the supply function satisfies \cref{4_eq:supplystability}, but with strict inequality when $\gpdot\neq 0$, and the system is observable, then the nonlinear system is USAS.
\end{theorem}
\begin{proof}
See \cref{4_pf:velostab}.
\end{proof}

In a similar fashion as for $\qsr$-USD, also for $\qsr$-VD, the condition in \cref{4_lem:velostab} reads as $\supR\preceq 0$. 


\subsection{Induced universal shifted dissipativity}\label{4_sec:veloshiftdissip}
In the previous section, we have seen how $\qsr$-VD implies US(A)S. Next, we are interested if $\qsr$-VD also implies $\qsr$-USD. Therefore, we formulate the following \lcnamecref{4_prop:veloshifteddissip}:
\begin{conjecture}[Induced $\qsr$-USD]\label{4_prop:veloshifteddissip}
	If a nonlinear system given by \cref{4_eq:nonlinsys} is $\qsr$-VD, then, it is also $\qsr$-USD for the same tuple $\qsr$.
\end{conjecture}

Due to technical reasons, can only provide a proof for a restricted case of \cref{4_prop:veloshifteddissip}. Nevertheless, later, we will show strong empirical evidence that \cref{4_prop:veloshifteddissip} does hold in practice.


Let us consider nonlinear systems given by
\begin{subequations}\label{4_eq:nonlinsysState}
\begin{align}
	\dot \st(t) &= \stMap(\st(t))+\ltiB \gd(t);\\
	\op(t) &= \ltiC \st(t).
\end{align}
\end{subequations}
Note that, at the cost of increasing the state dimension, e.g., by appending the system with appropriate input-output filters (see \cref{sec:app-coarse-sys}), we can always transform nonlinear systems given by \cref{4_eq:nonlinsys} to the form \cref{4_eq:nonlinsysState}. Besides considering systems of the form \cref{4_eq:nonlinsysState}, we will assume, in this subsection, that $\stSet$, i.e., the state set, is convex and compact.


Under these considerations, we will connect $\qsr$-VD for $\qsr$ tuples for which $\supS=0$, $\supQ\succeq 0$, and $\supR\preceq 0$ to USP notions that can be characterized by a similar US $\qsr$ supply function. Before presenting these results, we will first introduce two assumptions:
\begin{assumption}\label{4_as:CB}
	For the nonlinear system given by \cref{4_eq:nonlinsysState}, assume that $\ltiC \ltiB=0$.
\end{assumption}
\begin{proposition}\label{4_as:veloShiftBound2}
	Given a matrix $\supR\in\sym^\gpSize$ with $\supR\preceq 0$, there exists an $\alpha >0$ such that for all $\stEq\in\stSetEq$ and $\st\in\stSet$
	\begin{multline}\label{4_eq:veloshiftbound2}
	(\st-\stEq)^\top\intA(\st,\stEq)\!^\top \ltiC^\top \supR \ltiC\intA(\st,\stEq)(\st-\stEq)\leq\\\alpha^{-1}(\st-\stEq)^\top \ltiC^\top \supR \ltiC(\st-\stEq),
\end{multline}
where $\intA(\st,\stEq) = \int_0^1 \velA(\stEq+\var(x-\stEq))\, d\var$.
\end{proposition}
Assuming that $\velA$ is bounded, there will always exist an $\alpha$ for a given $\supR$ such that \cref{4_eq:veloshiftbound2} holds, as $\stSet$ is considered compact.

Similar to other works \cite{Koroglu2007,Wieland2009}, we assume that the disturbance $w$ is generated by an exosystem:
\begin{assumption}\label{4_as:wExoSys}
	For a given $(\stEq,\gdEq,\gpEq)\in\eqSet$, assume that $\gd$ is generated by the exosystem
	\begin{equation}\label{4_eq:exoSys}
		\dot \gd(t) = \exoA(\gd(t)-\gdEq),
	\end{equation}
	where $\exoA\in\reals^{\gdSize\times\gdSize}$ is Hurwitz and $\norm{\exoA}_{2,2}\leq \beta$. Define the corresponding behavior for a given $\gdEq \in \gdSetEq $ as
	\begin{equation}
		\exoBvr_{\gdEq}=\left\{ \gd\in \gdSet^{\nnreals} \mid \dot \gd\in\C{1},\,\gd\text{ satisfies } \cref{4_eq:exoSys}\right\}.
	\end{equation}
\end{assumption}
Note that constant and decaying (towards $\gdEq$) disturbances satisfy the behavior considered in \cref{4_as:wExoSys}.

Under these assumptions, we can formulate the following result to link $\qsr$-VD to USP:
\begin{theorem}[USP from $\qsr$-VD]\label{4_thm:veloshiftperf}
Consider a nonlinear system given by \cref{4_eq:nonlinsysState} for which \cref{4_as:CB,4_as:wExoSys} hold. If the system is $\qsr$-VD where $\supS=0$, $\supQ\succeq 0$, and $\supR\preceq 0$, then	
\begin{equation}\label{4_eq:pf:qrvsp}
	\int_{0}^{T} \beta^2(\star)^\top \supQ (\gd(t)-\gdEq)+\alpha^{-1}(\star)^\top  \supR (\gp(t)-\gpEq)\,dt \geq 0,
\end{equation}
for all $T> 0$, $(\gd,\gp)\in\proj_\mr{\gd,\gp}\Bc$ with $\gd\in\exoBvr_{\gdEq}$ and $\dot\st(0)=0$, and for every $(\stEq,\gdEq,\gpEq)\in\eqSet$.
\end{theorem}
\begin{proof}
	See \cref{4_pf:veloshiftperf}.
\end{proof}
 
Applying the result of \cref{4_thm:veloshiftperf} to the $\qsr$ tuple $\qsr = (\perf^2I,0,-I)$, corresponding to the (US) \ltwo-gain, we obtain the following \lcnamecref{4_cor:veloshiftl2}:
\begin{corollary}[Bounded \lstwo-gain from VD]\label{4_cor:veloshiftl2}
Consider a nonlinear system given by \cref{4_eq:nonlinsysState} for which \cref{4_as:CB,4_as:wExoSys} hold. If the system is $\qsr$-VD for $\qsr=(\perf^2I,0,-I)$,  where $\supR=-I$,  then the system has an \lstwo-gain bound of $\tilde\perf = \sqrt{\alpha\beta^2\perf^2}$.
\end{corollary}
\begin{proof}
See \cref{4_pf:veloshiftl2}.
\end{proof}

Note that these results form a direct upperbound for \cref{4_prop:veloshifteddissip}. The results of \cref{4_sec:veloshiftstab,4_sec:veloshiftdissip} constitute \cref{con:velo}. 





\section{Convex Universal Shifted Analysis}\label{4_sec:velousinglpv}
We have shown how $\qsr$-VD, considering a quadratic storage function $\storfunVelo$ of the form \cref{4_eq:velostorquad}, can be analyzed through a feasibility check of an infinite-dimensional set of LMIs. In this section, we will discuss how we can make the analysis computationally feasible and efficient through the use of methods from the LPV framework.

As the state-space matrices of the velocity form vary with $(\st,\gd)$, we obtain an infinite-dimensional set of LMIs for VD analysis. This is similar to the LPV case, where the state-space matrices in an LPV representation vary with the scheduling-variable $\sch$, which also results in an infinite-dimensional set of LMIs for classical dissipativity analysis of LPV systems \cite{Briat2015,Hoffmann2015}. In the LPV framework, various methods exist to turn the infinite-dimensional problem into a finite-dimensional one, such as through polytopic, grid-based, and multiplier-based techniques \cite{Hoffmann2015}. Inspired by the connection between the velocity form and LPV representations, we propose the use of LPV analysis results to make the VD analysis problem of nonlinear systems computationally feasible and efficient.

For that purpose, we introduce the following so-called \emph{Velocity Parameter-Varying} (VPV) embedding:
\begin{definition}[VPV embedding]\label{4_def:vpvembed}
	Consider a system given by \cref{4_eq:nonlinsys} with velocity form \cref{4_eq:veloform}. Furthermore, consider the LPV state-space representation
\begin{subequations}\label{4_eq:vpv}
		\begin{align}
		\dot\stdot(t)&=\lpvA(\sch(t))\stdot(t) + \lpvB(\sch(t))\gddot(t);\\ 
		\gpdot(t) &= \lpvC(\sch(t))\stdot(t) + \lpvD(\sch(t))\gddot(t);
	\end{align}
\end{subequations}
where $\stdot(t)\in\reals^\stSize$ is the state, $\gddot(t)\in\reals^\gdSize$ the input, and $\gpdot\in\reals^\gpSize$ the output of the LPV representation, with $\sch(t)\in\schSet\subset\reals^\schSize$ being the scheduling-variable, and matrix functions $\lpvA,\dots,\lpvD$ being of appropriate size and belonging to a given class of functions $\lpvClass$ (e.g., affine or rational). The LPV representation\footnote{Various methods available in the LPV framework to reduce the conservatism of the embedding, for a given dependency class of $\lpvA,\,\dots,\,\lpvD$ (e.g., affine, polynomial, rational, etc.), can also be used to reduce the conservatism of the VPV embedding, see e.g., \cite{Sadeghzadeh2020a,Toth2010}.} \cref{4_eq:vpv} is a so-called VPV embedding of \cref{4_eq:nonlinsys} on the region $\stSetLPV\times\gdSetLPV\subseteq \stSet\times\gdSet$, if there exists a function $\schMap : \stSetLPV\times\gdSetLPV \to \schSet$, the so-called scheduling-map, with $\sch = \schMap(\st,\gd)$, and $\schSet\supseteq\schMap(\stSetLPV,\gdSetLPV)$, such that:
\begin{equation}\label{4_eq:matvpvembed}
\begin{aligned}
\lpvA(\schMap(\st,\gd)) &= \velA(\st,\gd), \qquad \lpvB(\schMap(\st,\gd)) =  \velB(\st,\gd),\\
\lpvC(\schMap(\st,\gd)) &=  \velC(\st,\gd), \qquad \lpvD(\schMap(\st,\gd)) =  \velD(\st,\gd),
\end{aligned}
\end{equation}
for all $(\st,\gd) \in \stSetLPV\times\gdSetLPV$.
\end{definition}
The behavior of a VPV embedding given by \cref{4_eq:vpv} for a $\sch\in\schSet^\nnreals$ is
\begin{multline}
	\Blpv = \{(\stdot,\gddot,\gpdot)\in(\reals^\stSize\times\reals^\gdSize\times\reals^\gpSize)^{\nnreals}\mid \\\stdot\in\C{2},\, \gddot,\gpdot\in\C{1}\text{ and } (\stdot,\gddot,\gpdot,\sch)\text{ satisfy \cref{4_eq:vpv}}\},
\end{multline}
with $\Blpvfull = \bigcup_{\sch\in\schSet^\nnreals} \Blpv$ being the full behavior (i.e., for all scheduling trajectories).

By the VPV embedding, we have the following relation:
\begin{lemma}[VPV behavioral embedding]\label{4_lem:vpvembed}
Consider a nonlinear system given by \cref{4_eq:nonlinsys} with a velocity form \cref{4_eq:veloform}. If the LPV representation \cref{4_eq:vpv} is a VPV embedding of the nonlinear system on the region $\stSetLPV\times\gdSetLPV= \stSet\times\gdSet$, then the behavior of the velocity form is included in that of the LPV representation, i.e., $\Bv \subseteq \Blpvfull$.
\end{lemma}
\begin{proof}
	See \cref{4_pf:vpvembed}.
\end{proof}
This is similar to standard LPV embeddings \cite{Toth2010}. However, an important difference is that VPV embeddings embed the velocity form \cref{4_eq:veloform} while standard LPV embeddings embed the original nonlinear dynamics \cref{4_eq:nonlinsys}.
\begin{remark}
In the case that our VPV embedding only considers part of the state-space, i.e., \cref{4_eq:vpv} is a VPV embedding of the nonlinear system given by \cref{4_eq:nonlinsys} on the region $\stSetLPV\times\gdSetLPV\subset \stSet\times\gdSet$, we can still describe part of the behavior of the velocity form \cref{4_eq:veloform}. 
\end{remark}

Recall that VD of a nonlinear system, see \cref{4_def:velodissip}, can be seen as `classical dissipativity' of the velocity form. Through the VPV embedding, we can then cast the VD analysis problem as a classical dissipativity analysis problem of an LPV representation:
\begin{theorem}[VD analysis through the LPV framework]\label{4_thm:velodissiplpv}
	Consider the nonlinear system given by \cref{4_eq:nonlinsys} for which the LPV representation \cref{4_eq:vpv} is a VPV embedding of the system on the region $\stSetLPV\times\gdSetLPV= \stSet\times\gdSet$. If the LPV representation \cref{4_eq:vpv} is classically dissipative, then the system is VD.
\end{theorem}
\begin{proof}
See \cref{4_pf:velodissiplpv}.
\end{proof}

With \cref{4_thm:velodissiplpv}, we now have a powerful tool to analyze VD of nonlinear systems, as we can cast it as a classical dissipativity analysis problem of an LPV representation, for which there exist many results, see e.g. \cite{Hoffmann2015} for an overview. In \cref{4_sec:veloanalysis}, we have shown how US(A)S and quadratic USP notions can be analyzed through $\qsr$-VD. Consequently, combining these two results, we are now able to analyze US(A)S and USP of nonlinear systems through the use of the LPV framework. This then gives us a systematic, and computationally efficient tool to analyze global stability and performance in the form of US(A)S and USP, which constitutes \cref{con:lpv}.

In the next section, we will use these analysis tools to develop a novel systematic controller synthesis method in order to ensure and shape US(A)S and USP.

\section{Convex Universal Shifted Controller Synthesis}\label{4_sec:contrsynthesis}

\subsection{Controller synthesis problem}
In order to provide a systematic approach to controller synthesis, we will consider a notion similar to the generalized plant concept also commonly used for LTI and LPV systems \cite{Doyle1983,Apkarian1995}. This will allow us to describe various control configurations and shape the closed-loop behavior of the plant and controller to ensure US(A)S and USP.


More concretely, for the controller synthesis problem in this section, we will consider nonlinear systems $\genplant$ of the form
\begin{subequations}\label{4_eq:genplant}
    \begin{align}
    \dot\stp(t) &= \stMap(\stp(t),\ipp(t))+\ltiBw \gd(t);\\
    \gp(t) &=  \opgpMap(\stp(t),\ipp(t))+\ltiDzw \gd(t);\\
    \opp(t) &= \opopMap(\stp(t),\ipp(t))+\ltiDyw\gd(t);
    \end{align} 
\end{subequations}
where again $\stp(t)\in \stpSet\subseteq \reals^{\stpSize}$ is the state with $\stp\in\C{1}$ and initial condition $\stp(0) = \stpIc\in\reals^{\stpSize}$ and where now $\gd(t) \in \gdSet\subseteq\reals^\gdSize$ and $\gp(t)\in\gpSet\subseteq\reals^\gpSize$ are called the generalized disturbance (consisting of references, disturbances, etc.) and generalized performance (consisting  of tracking errors, control efforts, etc.) channels, respectively. Moreover, we introduce the channels $\ipp(t)\in\ippSet\subseteq\reals^\ippSize$ and $\opp(t)\in\oppSet\subseteq\reals^\oppSize$, denoting the control input and measured output channel, through which the to-be-designed controller will interact with the plant. Furthermore, $\stMap:\stpSet\times\ippSet\to\reals^\stpSize$, $\opgpMap:\stpSet\times\ippSet\to\reals^\gpSize$ and  $\opopMap:\stpSet\to\reals^\oppSize$ are assumed to be in $\C{1}$ and $\ltiBw\in\reals^{\stpSize\times\gdSize}$, $\ltiDzw\in\reals^{\gpSize\times\gdSize}$, and $\ltiDyw\in\reals^{\oppSize \times\gdSize}$. 


The to-be-designed controller $\controller$ for our plant $\genplant$ is considered to be of the form
\begin{subequations}\label{4_eq:generalcontroller}
\begin{align}
\dot\stk(t) &= \stkMap(\stk(t),\ipk(t));\\
\opk(t) &= \opkMap(\stk(t),\ipk(t));
\end{align}
\end{subequations}
where $\stk(t)\in\reals^\stkSize$ is the state, $\ipk(t)\in\reals^\ipkSize$ is the input, and $\opk(t)\in\reals^\opkSize$ is the output of the controller. Furthermore, $\stkMap:\reals^\stkSize\times\reals^\ipkSize\to\reals^\stkSize$ and $\opkMap:\reals^\stkSize\times\reals^\ipkSize\to\reals^\opkSize$. 

The closed-loop interconnection of $\genplant$ and $\controller$ for which $\ipk = \opp$ and $\ipp = \opk$ (hence, $\ipkSize=\oppSize$ and $\opkSize = \ippSize$) will be denoted by $\ic{\genplant}{\controller}$. Note that this closed-loop interconnection will be a system of the form \cref{4_eq:nonlinsys}, which has input $\gd$ and output $\gp$. Furthermore, the output $\gp\in\gpSet^\nnreals$ of $\ic{\genplant}{\controller}$ for an input $\gd\in\gdSet^\nnreals$ and initial condition $\stclIc = \col(\stp(0),\stk(0))\in\stpSet\times\reals^\stkSize$, will be denoted by $\ic{\genplant}{\controller}(\gd,\stclIc)=\gp\in\gpSet^\nnreals$.

Considering this closed-loop interconnection $\ic{\genplant}{\controller}$, our objective then is to synthesize the controller $\controller$ such that $\ic{\genplant}{\controller}$ is US(A)S and satisfies a (desired) USP criteria. To simplify the discussion, we will consider the \lstwo-gain, see \cref{4_eq:shiftlplqgain}, as our desired performance metric, which we aim to minimize. This means that we are interested in synthesizing a controller $\controller$ for our generalized plant $\genplant$ s.t. the \lstwo-gain $\perf$ of our closed-loop interconnection $\ic{\genplant}{\controller}$ from $\gd$ to $\gp$ is minimized, see also \cref{4_def:shiftlplqgain}. Other performance metrics can similarly be considered, for which similar guarantees can be derived as for the \lstwo-gain. 

%
	
To ensure that the given synthesis problem is feasible, we require $\genplant$ to be a generalized plant in the following sense:
\begin{definition}[Generalized plant]\label{4_def:genplant}
	$\genplant$, given by \cref{4_eq:genplant}, is a generalized plant, if there exists a controller $\controller$ of the form \cref{4_eq:generalcontroller} such that the closed-loop interconnection $\ic{\genplant}{\controller}$ is USS.
\end{definition}
\begin{proposition}\label{4_prop:genplant}
	$\genplant$, given by \cref{4_eq:genplant}, is a generalized plant in the sense of \cref{4_def:genplant}, if $\left(\Partial{\stMap}{\stp}(\stp,\ipp),\Partial{\stMap}{\ipp}(\stp,\ipp)\right)$ is stabilizable and $\left(\Partial{\stMap}{\stp}(\stp,\ipp),\Partial{\opopMap}{\stp}(\stp,\ipp)\right)$ is detectable over $\stpSet\times\gdSet$, see \cite[Section 5.3.2]{Pavlov2006}.
\end{proposition}
Note that \cref{4_prop:genplant} can be interpreted as the velocity form of \cref{4_eq:genplant} being stabilizable and detectable w.r.t. the input channel $\dot{\ipp}$ and output channel $\dot{\opp}$, respectively, along all trajectories of the nonlinear system given by \cref{4_eq:genplant}. This is similar to the condition that is required for standard LTI and LPV controller synthesis using the generalized plant concept.

\subsection{Universal shifted controller synthesis procedure}\label{4_sec:shiftcontr}
\subsubsection*{Overview}
We propose a novel procedure to synthesize a controller $\controller$ for our generalized plant $\genplant$ s.t. the closed-loop interconnection $\ic{\genplant}{\controller}$ is USAS with a bounded $\lstwo$-gain:
\begin{enumerate}
    \item \emph{VPV embedding step:}\label{4_step1} For a nonlinear generalized plant $\genplant$ given by \cref{4_eq:genplant}, compute its velocity form $\velogenplant$. Construct a VPV embedding $\velogenplantLPV$ of $\genplant$ based on $\velogenplant$.
    \item \emph{Velocity controller synthesis step:} For the VPV embedding $\velogenplantLPV$, an LPV controller $\velocontroller$ is synthesized, ensuring a minimal closed-loop \ltwo-gain $\perf$. For this, we use standard LPV controller synthesis methods.    
    \item \emph{Universal shifted controller realization step:} The synthesized controller $\velocontroller$, which is in the velocity domain, is realized as a nonlinear controller $\controller$ in the primal form \cref{4_eq:generalcontroller} to be used with the primal form of the generalized plant $\genplant$, to ensure the closed-loop \lstwo-gain $\perf$. 
\end{enumerate}

Note that in order to claim closed-loop performance guarantees in Step 3, we rely on \cref{4_prop:veloshifteddissip}, i.e., all typical $\qsr$ performance metrics can also be considered in the velocity controller synthesis in Step 2 to induce various quadratic USP notions of the closed-loop interconnection. While, under some technical assumptions, we presented a proof of \cref{4_prop:veloshifteddissip} for the \lstwo case, formal proof of the general $\qsr$-case is currently an open problem, although empirical evidence strongly suggests it holds in practice, as it is demonstrated in \cref{4_sec:examples}. Nonetheless, note that US(A)S \emph{is} guaranteed by \cref{4_thm:velotoshiftstab} and is therefore \emph{not} a conjecture.

\begin{remark}
	Note that the steps proposed above conceptually can be seen similar to those proposed in \cite{Koelewijn2020b}. However, in Step 1 and 2, we use a velocity form of the nonlinear system instead of a differential one, and hence, in Step 3, the controller realization procedure and the implied stability and performance guarantees are different and novel (constituting \cref{con:syn}). A serious advantage of the procedure proposed in this paper is that the resulting controller has a simpler structure and does not require explicit knowledge of the equilibrium points, which is an advantage compared to the controller realization in \cite{Koelewijn2020b}.
\end{remark}

Before detailing the individual steps of the above procedure, we will first show that the velocity form of the closed-loop interconnection is equal to the closed-loop interconnection of the velocity form of the plant and velocity form of the controller. This simplifies the controller design procedure, as it allows us to independently `transform' the plant and controller between their primal and velocity forms.
\begin{theorem}[Closed-loop velocity form] \label{4_thrm:veloic} The velocity form of the closed-loop system $\ic{\genplant}{\controller}$ is equal to the closed-loop interconnection of $\velogenplant$ and $\velocontroller$, i.e., $\ic{\velogenplant}{\velocontroller}$, if the interconnection of $\genplant$ and $\controller$ is well-posed, i.e., there exists a $\C{1}$ function $\breve{\opMap}$ such that $\ipp = \opkMap(\stk,\opopMap(\stp,\ipp))$ can be expressed as
		$\ipp = \breve{\opMap}(\stp,\stk)$.    
\end{theorem}
\begin{proof} 
The proof follows in a similar manner as the proof of Theorem 5 in \cite{Koelewijn2020b}.
\end{proof}

Now, we detail each step of the proposed procedure.

\subsubsection{VPV embedding step}
The first step of our US controller synthesis procedure consists of embedding the generalized plant $\genplant$ given by \cref{4_eq:genplant} in a VPV embedding. From here on, we will denote the behavior of \cref{4_eq:genplant} by
\begin{multline}\label{4_eq:nonlinsolgen}
\Bp = \Big\{ (\stp,\ipp,\gd,\gp,\opp)\in (\stpSet\times \ippSet\times\gdSet\times\gpSet\times\oppSet)^\nnreals \mid \\\stp\in\C{1}, (\stp,\ipp,\gd,\gp,\opp)\text{ satisfies \cref{4_eq:genplant}}\Big\},
\end{multline}
and the set of all differentiable solutions as 
\begin{equation}
	\Bc = \{(\stp,\ipp,\gd,\gp,\opp) \in\Bp \mid \stp \in\C{2},\, \ipp,\gd,\gp,\opp\in\C{1}\}.
\end{equation}

First, we compute the velocity form of $\genplant$, resulting in
\begin{subequations}\label{4_eq:velogenplant}
    \begin{align}
    \ddot\stp &= \velA(\stp,\ipp)\dot\stp +\ltiBw \dot\gd + \velBu(\stp,\ipp)\dot\ipp;\\
    \dot\gp &=  \velCz(\stp,\ipp)\dot\stp +\ltiDzw \dot\gd + \velDzu(\stp,\ipp)\dot\ipp ;\\
    \dot\opp &= \velCy(\stp,\ipp)\dot\stp+\ltiDyw \dot\gd + \velDyu(\stp,\ipp)\dot\ipp;
    \end{align} 
\end{subequations}
where $\velA = \Partial{\stMap}{\stp}$, $\velB = \Partial{\stMap}{\ipp}$, $\velCz = \Partial{\opgpMap}{\stp}$, $\velDzu = \Partial{\opgpMap}{\ipp}$, $\velCy = \Partial{\opopMap}{\stp}$, and $\velDyu = \Partial{\opopMap}{\ipp}$. The behavior of \cref{4_eq:velogenplant} is then $\Bv=\dotB\Bc$, see also \cref{4_sec:veloform}. 

We then embed $\velogenplant$, given by \cref{4_eq:velogenplant}, in an LPV representation in order to construct a VPV embedding of $\genplant$ \cref{4_eq:genplant}. Based on \cref{4_def:vpvembed} of the VPV embedding, we then construct a VPV embedding of $\genplant$ given by \cref{4_eq:genplant} on the compact region\footnote{Note that as $\gd$ enters into \cref{4_eq:genplant} linearly, the scheduling-map $\schMap$ of the VPV embedding will not depend on it. Hence, the embedding region w.r.t. $\gd$ can be taken equal to or as any subset of $\gdSet$, i.e., the complete value set of $\gd$. Therefore, we omit it when discussing the VPV embedding region of \cref{4_eq:genplant}.} $\stpSetLPV\times\ippSetLPV\subseteq\stpSet\times\ippSet$, which we denote by $\velogenplantLPV$ and is given by
\begin{subequations}\label{4_eq:vpvgenplant}
    \begin{align}
    \dot\stpdot &= \lpvA(\sch)\stpdot +\ltiBw \gddot + \lpvBu(\sch)\ippdot;\\
    \gpdot &=  \lpvCz(\sch)\stpdot +\ltiDzw \gddot + \lpvDzu(\sch)\ippdot ;\\
    \oppdot &= \lpvCy(\sch)\stpdot+\ltiDyw \gddot + \lpvDyu(\sch)\ippdot;
    \end{align} 
\end{subequations}
with scheduling-variable $\sch(t)\in\schSet\subset\reals^\schSize$, where $\schSet$ is assumed to be convex, and the accompanying scheduling-map $\schMap$ is s.t. $\sch(t) = \schMap(\stp(t),\ipp(t))$ and $\schSet\supseteq \schMap(\stpSetLPV,\ippSetLPV)$. We will also assume that $\schMap\in\C{1}$.  Moreover, $\stpdot(t)\in\reals^\stpSize$, $\gddot(t)\in\reals^\gdSize$, $\ippdot(t)\in\reals^\ippSize$, $\gpdot(t)\in\reals^\gpSize$, and $\oppdot(t)\in\reals^\oppSize$. The accompanying behavior of the VPV embedding \cref{4_eq:vpvgenplant} for a scheduling trajectory $\sch\in\schMap^\nnreals$ is 
\begin{multline}
	\Blpv = \{(\stpdot,\ippdot,\gddot,\gpdot,\oppdot)\in(\reals^\stpSize\times\reals^\ippSize\times\reals^\gdSize\times\\\reals^\gpSize\times\reals^\oppSize)^{\nnreals}\mid \stpdot\in\C{2},\, \ippdot,\gddot,\gpdot,\oppdot\in\C{1}\text{ and}\\ (\stpdot,\ippdot,\gddot,\gpdot,\oppdot,\sch)\text{ satisfy \cref{4_eq:vpvgenplant}}\},
\end{multline}
with $\Blpvfull = \bigcup\limits_{\sch\in\schSet^\nnreals} \Blpv$ being the full behavior of \cref{4_eq:vpvgenplant}.

Moreover, we will denote the restriction of the state and control input solutions of $\genplant$ to $\stSetLPV$ and $\ipSetLPV$, respectively, by $\BcXU= \{(\stp,\ipp,\gd,\gp,\opp) \in\Bc \mid (\stp(t),\ipp(t))\in\stpSetLPV\times\ippSetLPV\}$ and the corresponding solution set for the velocity form $\velogenplant$ by $\BvXU=\dotB\BcXU$. Through the VPV behavioral embedding principle, given by \cref{4_lem:vpvembed}, we have $\BvXU\subseteq\Blpvfull$. This means that through $\velogenplantLPV$ given by \cref{4_eq:vpvgenplant}, we can describe the behavior of $\velogenplant$ given by \cref{4_eq:velogenplant} for which $(\stp(t),\ipp(t))\in\stpSetLPV\times\ippSetLPV$.

\begin{remark}
For the controller synthesis problem to be feasible, one must make sure that for the VPV embedding given by \cref{4_eq:vpvgenplant}, the pairs $(\lpvA, \ltiBu)$ and $(\lpvA, \ltiCy)$ are stabilizable and detectable, respectively \cite{Wu2006}. This is necessary to preserve the stabilizability and detectability properties of the underlying velocity form according to \cref{4_prop:genplant}.
\end{remark}

\subsubsection{Velocity controller synthesis step}
Having constructed a VPV embedding $\velogenplantLPV$ for our generalized plant $\genplant$, we will use it to synthesize a controller $\velocontroller$ for $\velogenplant$ s.t. $\ic{\velogenplant}{\velocontroller}$ has a minimal \ltwo-gain. 
As $\velogenplantLPV$ is an LPV representation, we make use of the synthesis algorithms in the LPV framework to synthesize an LPV controller $\velocontroller$ s.t. $\ic{\velogenplantLPV}{\velocontroller}$ has a minimal \ltwo-gain. This will ensure that the \ltwo-gain of $\ic{\velogenplant}{\velocontroller}$ is minimized. Note that this is a standard LPV synthesis problem, hence, we can apply one of the various available controller synthesis techniques that ensure the closed-loop interconnection is classically dissipative with a minimal \ltwo-gain bound, see e.g. \cite{Apkarian1995,Packard1993,Scherer2001,Wu1995}.

 Concretely, we consider $\velocontroller$, which we will refer to as velocity controller, in the form
\begin{subequations}\label{4_eq:veloContr}
    \begin{align}
    \dot\stkdot(t) &= \lpvAk(\sch(t)) \stkdot(t) + \lpvBk(\sch(t)) \ipkdot(t);\label{4_eq:velContrSt}\\
    \opkdot(t) &= \lpvCk(\sch(t)) \stkdot(t) + \lpvDk(\sch(t)) \ipkdot(t);\label{4_eq:velContrOp}
    \end{align}
\end{subequations}
where $\stkdot(t) \in \reals^{\stkSize}$ is the state, $\ipkdot(t) \in \reals^{\ipkSize}$ is the input, and $\opkdot(t) \in \reals^{\opkSize}$ is the output of the controller, respectively, and $\lpvAk,\ldots,\lpvDk\in\lpvClass$ are matrix functions with appropriate dimensions. Note that when we connect $\velocontroller$ to $\velogenplant$ (to obtain $\ic{\velogenplant}{\velocontroller}$), we have that $\ipkdot = \dot\opp$ and $\opkdot = \dot\ipp$, and $\sch=\schMap(\stp,\ipp)$. Moreover, as $\stp,\ipp\in\C{1}$ and $\schMap\in\C{1}$, we have $\sch\in\C{1}$.

Based on this, we can formulate the following \lcnamecref{4_thrm:veloICL2}:
\begin{theorem}[Velocity closed-loop \ltwo-gain]\label{4_thrm:veloICL2}
	If controller $\velocontroller$ of the form \cref{4_eq:veloContr} ensures classical dissipativity and a bounded \ltwo-gain $\perf$ of the closed-loop interconnection $\ic{\velogenplantLPV}{\velocontroller}$ for all $(\stpdot,\ippdot)\in\proj_{\stpdot,\ippdot}\Blpvfull$, then, 
	$\ic{\velogenplant}{\velocontroller}$ with $p=\schMap(\stp,\ipp)$ is classically dissipative and has an \ltwo-gain bound 
	$\leq\perf$ for all $(\dot\stp,\dot\ipp)\in\proj_\mr{\dot\stp,\dot\ipp}\BvXU$.
\end{theorem}
\begin{proof}
See \cref{4_pf:veloICL2}.
\end{proof}

\begin{remark}\label{4_rem:weightfilt}
By applying shaping filters on $\genplant$ that consequently appear in $\velogenplant$, we can shape the closed-loop performance of $\ic{\genplant}{\controller}$, see \cref{4_fig:shapefigprim}. If the weighting filters included in $\genplant$ are LTI, then the input-output behavior of $\weight_\mr{\gd}$ and $\weight_\mr{\gp}$ is equivalent to that of $\weight_\mr{v,\gd}$ and $\weight_\mr{v,\gp}$. This is because the transfer function representation of the velocity form of an LTI system is given by the same transfer function as its primal form. This results in a one-to-one correspondence between the performance shaping of the primal form $\ic{\genplant}{\controller}$ (see \cref{4_fig:shapefigprim}) and performance shaping of the velocity form $\ic{\velogenplant}{\velocontroller}$ (see \cref{4_fig:shapefigvelo}). This significantly simplifies the controller design, as shaping can be directly performed through the velocity form $\velogenplant$ and hence also through the VPV embedding $\velogenplantLPV$.
\end{remark}

\begin{figure}
	\centering
	\begin{subfigure}[b]{\columnwidth}
		\centering
		\includegraphics[scale=1]{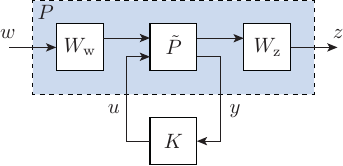}
		\caption{Primal form.}
		\label{4_fig:shapefigprim}
	\end{subfigure}\\[1em]
	\begin{subfigure}[b]{\columnwidth}
		\centering
		\includegraphics[scale=1]{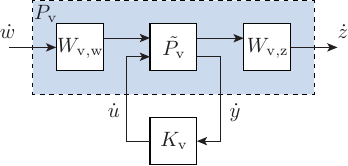}
		\caption{Velocity form.}
		\label{4_fig:shapefigvelo}
	\end{subfigure}
	\caption{{Shaping the closed-loop behavior of the primal and the velocity form by the use of weighting filters $W_\mr{w}$ and $W_\mr{z}$.}}
	\label{4_fig:shapefig} 
	\vspace{-1em}
\end{figure}

\subsubsection{Universal shifted controller realization step}
Finally, we will describe the last step of the proposed synthesis procedure. For the last step, we realize the controller $\controller$ to be used with the primal form of the generalized plant $\genplant$ based on the velocity controller $\velocontroller$ of the previous step, such that closed-loop US(A)S and USP are ensured.

To do this, we exploit the properties of the velocity form. As in the velocity form, the inputs and outputs of our system are time-differentiated versions of the inputs and outputs of the primal form of the system. This allows us to formulate the following \lcnamecref{4_thrm:timediff}:
\begin{theorem}[Velocity behavior equivalence]\label{4_thrm:timediff}
Consider a nonlinear system given by \cref{4_eq:nonlinsys} with behavior $\Bc$ (see \cref{4_eq:Bc}). The nonlinear system with its inputs (time) integrated and its outputs (time) differentiated is equal to its velocity form \cref{4_eq:veloform}.
\end{theorem} 
\begin{proof}
See \cref{4_pf:timediff}.
\end{proof}

Next, we exploit the result of \cref{4_thrm:timediff} for our controller realization. Namely, we concatenate an integrator to the output of the velocity controller $\velocontroller$ and concatenate a differentiator to its input. We then absorb the differentiator and integrator into the dynamics of the controller. Let us denote the $i$'th element of $\sch$ by $\sch_i$.
\begin{theorem}[Universal shifted controller realization]\label{4_thrm:controlrealiz}
Consider the velocity controller $\velocontroller$ given by \cref{4_eq:veloContr}. Furthermore, consider the nonlinear time-invariant controller $\controller$ given by
\begin{subequations}\label{4_eq:Contr}
\begin{align}
	\dotstkus(t) &= \lpvAkus(\sch(t))\stkus+\lpvBkus(\sch(t),\dot\sch(t)) \ipk(t);\\
	\opk(t) &= \lpvCkus\stkus+\lpvDkus(\sch(t)) \ipk(t);
\end{align}
\end{subequations}
where $\stkus(t)\in\reals^{\stkSize+\opkSize}$ is the state of the controller, and where 
\begin{equation}\label{4_eq:shiftcontrmat}
	\begin{alignedat}{2}
		\lpvAkus(\sch) \!&= \!\begin{bmatrix}
			\lpvAk(\sch) & \!\!0\\\lpvCk(\sch) & \!\!0
		\end{bmatrix}\!,\, &&\lpvBkus(\sch,\dot\sch) \!=\! \begin{bmatrix}
			\lpvAk(\sch)\lpvBk(\sch)-\dif\lpvBk(\sch,\dot\sch)\\
			\lpvCk(\sch)\lpvBk(\sch)-\dif\lpvDk(\sch,\dot\sch)
		\end{bmatrix}\!,\\
		\lpvCkus \!&=\! \begin{bmatrix}
			0 & I
		\end{bmatrix}, &&\lpvDkus(\sch) \!=\! \lpvDk(\sch),
	\end{alignedat}
\end{equation}
with (omitting time dependence) $\sch = \schMap(\stp,\ipp)$, $\dif\lpvBk(\sch,\dot\sch) = \sum_{i=1}^\schSize\Partial{\lpvBk(\sch)}{\sch_i}\dot\sch_i$, and $\dif\lpvDk(\sch,\dot\sch) = \sum_{i=1}^\schSize\Partial{\lpvDk(\sch)}{\sch_i}\dot\sch_i$. The controller $\controller$ in \cref{4_eq:Contr} is the primal form of $\velocontroller$ \cref{4_eq:veloContr} and the velocity form of $\controller$ is $\velocontroller$. Hence, $\controller$ is called the primal realization of $\velocontroller$.
\end{theorem}
\begin{proof}
See \cref{4_pf:controlrealiz}.
\end{proof}
We will refer to the controller $\controller$ in \cref{4_eq:Contr} as the \emph{US controller}. For the realization of this controller, we require the derivative $\dot\sch$ of the scheduling-variable. Note that $\dot\sch$ exists, as $\sch\in\C{1}$, and is bounded as the VPV embedding region $\stpSetLPV\times\ippSetLPV$ is compact. Moreover, also note that the dependency on $\dot\sch$ in \cref{4_eq:Contr} drops out when $\lpvBk$ and $\lpvDk$ of the velocity controller \cref{4_eq:veloContr} are constant matrices. This then has to be ensured in Step 2 of the controller synthesis procedure. Note by incorporating a low-pass filter in $\velocontroller$ for roll-off in the high frequency range this can easily be ensured. Finally, note that the realization of the controller \cref{4_eq:Contr} is not necessarily state minimal. However, in the literature techniques exist which can be applied to construct a minimal state-space realization or perform state-reduction of an LPV system, see e.g. \cite{Petreczky2017}. 
An interpretation of this controller realization procedure is also depicted in \cref{4_fig:controlrel}.

\begin{figure}
	\centering
	\includegraphics{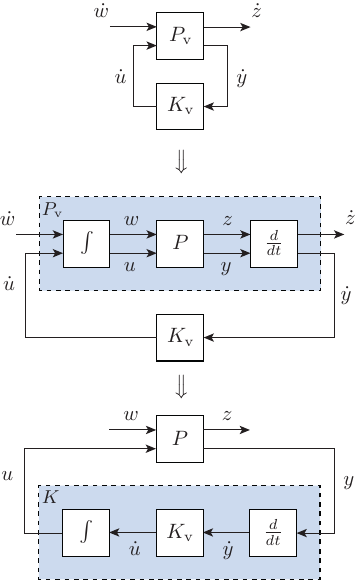}
	\caption{Universal shifted controller realization.}
	\label{4_fig:controlrel} 
	\vspace{-1em}
\end{figure}

\begin{remark}
Similar control structures as \cref{4_eq:Contr} have been proposed in literature, see e.g. \cite{Kaminer1995,Mehendale2006}. Compared to these works, we connect the proposed controller design to US(A)S and USP, which, to the authors' knowledge, has not been done so far in literature.
\end{remark}

\subsection{Closed-loop universal shifted stability and performance}\label{4_sec:clpperfstuffs}
Based on the proposed US controller realization of \cref{4_sec:shiftcontr}, we can show that the closed-loop interconnection $\ic{\genplant}{\controller}$ is USAS and has \lstwo-gain bounded by $\perf$ (under the consideration that \cref{4_prop:veloshifteddissip} is true). Before showing this result, we first introduce the following definition:
\begin{definition}[Invariance]\label{4_def:invariance}
	For a system\,\footnote{Note that $\ic{\genplant}{\controller}$ is of this form.} given by \cref{4_eq:nonlinsys} with behavior $\B$, we call $\tilde\stSet\subseteq\stSet$ to be invariant under a given $\tilde\gdSet\subseteq\gdSet$, if  $\st(t) = \sttran(t,0,\stIc,\gd) \in \tilde\stSet$ for all $t\in\nnreals$, $\stIc\in\tilde\stSet$ and $\gd\in\gdSet^\nnreals$. The corresponding behavior is denoted by $\Binv{\tilde\stSet\tilde\gdSet}=\B \cap \{(\st,\gd,\gp)\in\B \mid (\st(t),\gd(t))\in \tilde\stSet\times\tilde\gdSet,\,\forall\,t\in\nnreals\}$.
\end{definition}
\begin{theorem}[Closed-loop USAS and USP]\label{4_thrm:clpls2}
 Let $\velocontroller$, given in \cref{4_eq:veloContr}, be an LPV controller, synthesized for the velocity form $\velogenplant$ in \cref{4_eq:velogenplant} of a nonlinear system given by \cref{4_eq:genplant} with behavior $\Bc$, which ensures classical dissipativity and a bounded \ltwo-gain of $\perf$ of the closed-loop $\ic{\velogenplant}{\velocontroller}$ on $\stpSetLPV\times\ippSetLPV$.  Consider the set $\tilde\gdSet\subseteq \gdSet$, such that there is an open and bounded $\stkSet \subseteq \reals^\stkSize$ for which $\stclSet = \stSetLPV \times \stkSet $ is invariant in the sense of \cref{4_def:invariance}. Then, the controller $\controller$, given by \cref{4_eq:Contr}, ensures USAS and an \lstwo-gain of $\leq\perf$ of the closed-loop $\ic{\genplant}{\controller}$ for all $\gd\in{\tilde\gdSet}^\nnreals\cap\ltwoe$ and any $\gdEq\in\gdSetEq\cap\tilde\gdSet$.
 \end{theorem}
 \begin{proof}
 See \cref{4_pf:clpls2}.
\end{proof}
Note that considering $\gd\in{\tilde\gdSet}^\nnreals\cap\ltwoe$ in \cref{4_thrm:clpls2} ensures that the trajectories stay in the VPV embedding region considered during synthesis of the controller. However, computing $\tilde\gdSet$ is a difficult problem, which is related to reachability analysis or invariant set computation. Nonetheless, there are numerical tools that can be employed for this purpose, see e.g.~\cite{Althoff2013,Maidens2015}.

Therefore, we have shown that the proposed controller synthesis procedure allows us to realize a nonlinear US controller that ensures US(A)S and USP, which constitutes \cref{con:syn}.

\subsection{Reference tracking and disturbance rejection}\label{4_sec:refdistreal}
The previously presented US controller design makes use of the velocity form and VD to ensure US(A)S and USP of the closed-loop. This has the advantage that explicit knowledge of the (closed-loop) equilibrium points is not required, making the design procedure feasible. However, for reference tracking and disturbance rejection purposes, it is important that the controller is designed in such a way that the equilibrium points of the closed-loop system correspond to the to-be-followed (constant) reference trajectories.

To achieve this for the US controller design described in this paper, we propose the following solution. For reference tracking and disturbance rejection, the generalized disturbance channel $\gd$ of the generalized plant is often in the form $\gd = \col(\gd_1,\gd_2)$ with $\gd_1$ containing the reference signals and $\gd_2$ containing external disturbances. Furthermore, corresponding to this, the measured output of the plant is assumed to be of the form $\opp = \col(\opp_1,\opp_2)$, where $\opp_1$ contains signals to be tracked (such that $\gd_1$ and $\opp_1$ have the same dimension) and $\opp_2$ contains other to be controlled variables. As in controller design for LTI systems, an integrator is required to be included in the dynamics of the controller $\controller$ corresponding to the $\opp_1$ channel to track at least a constant reference.  A simple way to achieve explicit integral action is by including a bi-proper filter with integrator(s) in the loop, or approximate integral action can be achieved by appropriate choice of weighting filters, see \cite[Section 17.4]{Zhou1996}. Including an explicit integral filter results in the interconnection depicted in \cref{4_fig:contrAndInt}. See also \cref{4_fig:genplantintlooppp}, where a generalized plant $\genplant$ is depicted with an explicit integral filter in the loop (where $\opp_2$ is empty, i.e., $\opp_1=\opp$).

\begin{figure}
	\centering
	\includegraphics[scale=1]{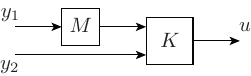}
	\caption{Controller interconnection with controller $\controller$ and integral filter $\weightint$ for reference and disturbance rejection.}
	\label{4_fig:contrAndInt}
	\vspace{-.5em}
\end{figure}

\begin{figure}
	\centering
	\includegraphics[scale=1]{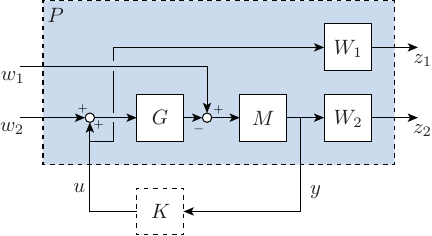}
	\caption{Example closed-loop control configuration with explicit integral filter $\weightint$ in the loop.}
	\label{4_fig:genplantintlooppp}
	\vspace{-1em}
\end{figure}

The inclusion of explicit integrators in the loop allows for state reduction of the interconnection of integrators and primal controller, as the differentiators used for realization of the controller (see \cref{4_sec:shiftcontr}) and the integrators in the loop partially cancel out. This yields the following \lcnamecref{4_cor:intrealiz}:
\begin{corollary}[Realization with integral action]\label{4_cor:intrealiz}
	Consider a generalized plant which includes an explicit integrator filter of the form $\weightint(s) = \frac{s+\alpha}{s}$ (where $s$ is the Laplace variable and $\alpha>0$) in the loop, such that the (to-be-designed) controller $\controller$ and $\weightint$ are connected as depicted in \cref{4_fig:contrAndInt} where $\opp_2$ is empty, i.e., $\opp_1=\opp$ (see also see \cref{4_fig:genplantintlooppp}). For $\velocontroller$ given by \cref{4_eq:veloContr}, the interconnection of the primal realization of the controller $\controller$ and $\weightint$ can be expressed as \cref{4_eq:Contr} where $\lpvAkus$, $\lpvCkus$, and $\lpvDkus$ are given as in \cref{4_eq:shiftcontrmat}, while $\lpvBkus$ is
	\begin{equation}\label{4_eq:Bcontr}
		\lpvBkus(t) = \begin{bmatrix}
			\lpvAk(\sch(t))\lpvBk(\sch(t))+\lpvBk(\sch)\alpha I -\dif\lpvBk(\sch(t),\dot\sch(t))\\
			\lpvCk(\sch(t))\lpvBk(\sch(t))+\lpvDk(\sch)\alpha I -\dif\lpvDk(\sch(t),\dot\sch(t))
		\end{bmatrix}.
	\end{equation}
\end{corollary}
\begin{proof}
See \cref{4_pf:intrealiz}.
\end{proof}

\section{Example\extver{s}{}}\label{4_sec:examples}
In this section, we  will demonstrate through \extver{examples}{an example}that the US controller design guarantees closed-loop USAS and \lstwo-gain performance. The results will be demonstrated through \extver{a simulation study and}{}experiments on a real system. \extver{}{A simulation study on a Duffing oscillator is also included in the extended version \cite{Koelewijn2023}.}Moreover, we compare the US controller design to a standard LPV controller design, which is only able to guarantee stability w.r.t. the origin and standard \ltwo-gain performance.

\extver{
\subsection{Duffing oscillator}\label{4_sec:duffspring}
First, constant reference tracking and disturbance rejection for a Duffing oscillator is investigated. The system is described by the following differential equations:
\begin{equation}\label{4_eq:duffNL}
\begin{aligned}
        \dot{q}(t) &= v(t);\\
        \dot{v}(t) &= -\frac{k_1}{m} q(t) -\frac{k_2}{m} \left(q(t)\right)^3 - \frac{d}{m} v(t) + \frac{1}{m} F(t);
       \end{aligned}
\end{equation}
where, $q$ $[\mr{m}]$ is the position, $v$ $[\mr{m\cdot s^{-1}}]$ the velocity and $F$ $[\mr{N}]$ is the (input) force acting on the mass. Furthermore, $m = 1\;[\mathrm{kg}]$, $k_1 = 0.5\;[\mathrm{N\cdot m^{-1}}]$, $k_2 = 5\;[\mathrm{N\cdot m^{-3}}]$ and $d = 0.2\;[\mathrm{N\cdot s\cdot m^{-1}}]$. We assume that only the position $q$ can be measured and hence it is considered to be the only output of the plant.

The generalized plant $\genplant$ considered for synthesis is depicted in \cref{4_fig:genplantscor}, where $\plant$ is the system given by \cref{4_eq:duffNL}, $\controller$ is the controller, and $\gd =\col\left( \rf ,\dist_\mathrm{i}\right)$ is the generalized disturbance, with $\rf$ the reference and $\dist_\mathrm{i}$ an input disturbance. The performance channel consists of $\gp_1$ (tracking error) and $\gp_2$ (control effort). The LTI weighting filters $\lbrace W_i \rbrace_{i = 1}^3$ used for performance shaping are $\weight_1(s) = \frac{0.501(s + 3)}{s + 2\pi}$, $\weight_2(s) = \frac{10(s+50)}{s + 5\cdot 10^4}$ and $\weight_3 = 1.5$. Furthermore, integral action is enforced by the filter $\weightint(s) = \frac{s+2\pi}{s}$. The resulting sensitivity weight $\weight_1(s)\weightint(s)$ has guaranteed 20 dB/dec roll-off at low frequencies to ensure good tracking performance, while $\weight_2(s)$ has high-pass characteristics to ensure roll-off at high frequencies.

\begin{figure}
    \centering
    \includegraphics[scale=1]{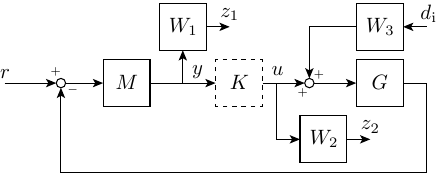}
    \caption{Generalized plant for the Duffing oscillator example.}
    \label{4_fig:genplantscor}
    \vspace{-1em}
\end{figure}

As the Duffing Oscillator is the only nonlinear system in the generalized plant, we only require the computation of the velocity form of \cref{4_eq:duffNL} for the VPV embedding, as all the LTI filters will have the same dynamics in their velocity forms. The following VPV embedding for \cref{4_eq:duffNL} is constructed\footnote{Note that variables with subscript $\mr{v}$ correspond to the time-derivatives, i.e., $\veloform{q}$ corresponds to $\dot{q}$ etc. See also \cref{4_def:vpvembed}.}, where dependence on time is omitted for brevity:
\begin{equation}\label{4_eq:duffvpv}
    \begin{aligned}
        \dot \veloform{q}   &= \veloform{v};\\
        \dot \veloform{v}  &=   \left( -\frac{k_1}{m}-3 \frac{k_2}{m} \sch \right)  \veloform{q} - \frac{d}{m} \veloform{v} + \frac{1}{m}\veloform{F};
    \end{aligned}
\end{equation}
where the scheduling $\sch(t) = q^2(t)$ and $\schSet=[0,\, 2]$ is chosen to allow for a relatively large operating range. Consequently, the resulting VPV embedding region is $\stSetLPV = [-\sqrt{2},\,\sqrt{2}]\times\reals$. 

The LPVcore Toolbox \cite{DenBoef2021} is used to synthesize a controller of the form \cref{4_eq:veloContr} for the velocity controller synthesis (i.e., second) step of the US controller design procedure \cref{4_sec:shiftcontr}. For this step, $\lpvBk$ and $\lpvDk$ are enforced to be constant, which gives us an \lstwo-gain of $\perf = 1.2$. Based on this structural restriction, the resulting US controller has affine dependence on $\sch(t)$ without dependence on $\dot \sch(t)$, see \cref{4_eq:Contr}. Moreover, as an integration filter $\weightint$ is included in the loop, we make use of the result of \cref{4_cor:intrealiz} for the primal realization to obtain the US controller. Based on this, the closed-loop is USAS and has an \lstwo-gain of $\perf \leq 1.2$ (under the consideration that \cref{4_prop:veloshifteddissip} is true) for $(q(t),v(t))\in [-\sqrt{2},\,\sqrt{2}]\times\reals$.

For comparison, a standard LPV controller design is also done to ensure \ltwo-gain performance and closed-loop stability (of the origin). For the design of the standard LPV controller, the primal form of the system given by \cref{4_eq:duffNL} is embedded in an LPV representation:
\begin{equation}\label{4_eq:duffLPV}
    \begin{aligned}
        \dot{q} &= v;\\
        \dot{v} &= \left(-\frac{k_1}{m}-\frac{k_2}{m}\sch_\mathrm{o}\right) q - \frac{d}{m} v + \frac{1}{m} F;
    \end{aligned}
\end{equation}
where $\sch_\mathrm{o}(t) = q^2(t)$ is the scheduling-variable. Here we will denote with subscript `o' the standard concept of LPV embedding and control design. Note that $\sch_\mr{o}$ is the same as $\sch$, i.e., we are able to create an LPV embedding with the same scheduling-map. Consequently, for comparison, we also consider $\sch_\mr{o}(t)\in\schSet_\mathrm{o}=[0,\, 2]$. The same generalized plant structure as for the US design is considered, see \cref{4_fig:genplantscor}, and we also use the LPVcore toolbox to synthesize the standard LPV controller. This then results in an \ltwo-gain for the standard LPV controller design of $\perf = 0.94$.

\begin{figure}
    \centering
    \includegraphics[scale=1]{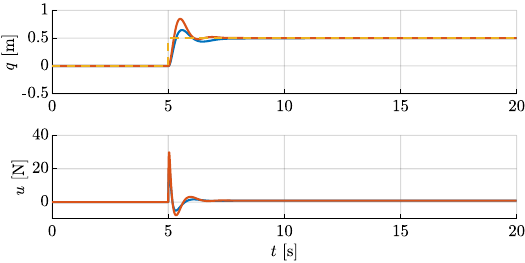}
    \caption{Position of the Duffing oscillator (top) in closed-loop with the standard LPV (\legendline{mblue}) and the US (\legendline{morange}) controllers under reference (\legendline{myellow,dashed}) and no input disturbance, together with the generated control inputs (bottom) by the controllers.}
    \label{4_fig:duffNoDist}
    \vspace{-1em}
\end{figure} 
\begin{figure}
    \centering
    \includegraphics[scale=1]{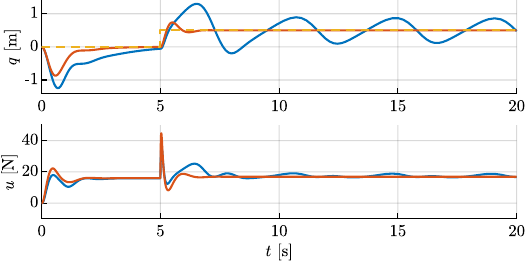}
    \caption{Position of the Duffing oscillator (top) in closed-loop with the standard LPV (\legendline{mblue}) and the US \mbox{(\legendline{morange})} controllers under reference (\legendline{myellow,dashed}) and input disturbance, together with the  generated control inputs (bottom) by the controllers.}
    \label{4_fig:duffDist}
    \vspace{-.5em}
\end{figure}
In simulation, the resulting outputs of the system using the standard LPV controller and the US controller in closed-loop are depicted without and with input disturbance in \cref{4_fig:duffNoDist,4_fig:duffDist}, respectively. In both cases, a step signal is taken as a reference trajectory which changes from zero to 0.5 at $t=5$ seconds. For the simulation results in \cref{4_fig:duffDist}, a constant input disturbance $\dist_\mr{i} \equiv -10\tfrac{2}{3}$ (corresponding to $-10\tfrac{2}{3}\cdot \weight_3=-16$ [N]) is applied. Note that this reference signal is not continuously differentiable, due to the change in value at 5 seconds. Nonetheless, as can be seen, the proposed US controller works for non-differentiable references as well.

 Comparing the results of the standard LPV controller and the US controller in \cref{4_fig:duffNoDist} shows that both controllers have similar performance when no input disturbance is present. The US controller has slightly more overshoot, but a lower settling time for this example. However, under constant input disturbance, which requires stabilization of a new set point, it can be seen in \cref{4_fig:duffDist} that the standard LPV controller has a significant performance loss with oscillatory behavior. This is because the LPV controller has no guarantees for stability and performance w.r.t. set points other than the origin. On the other hand, the proposed US controller ensures USAS and USP and therefore preserves the closed-loop performance properties under different setpoints, as visible in \cref{4_fig:duffDist}. Note, that in both cases, the scheduling-variable $\sch$ never leaves the set for which the controllers have been designed, i.e., $q(t) \in [-\sqrt{2},\,\sqrt{2}]$.
 }
 {}

\subsection{Unbalanced disk system}\label{4_sec:unbdisc}
\begin{figure}
	\centering
	\includegraphics[width=0.45\columnwidth]{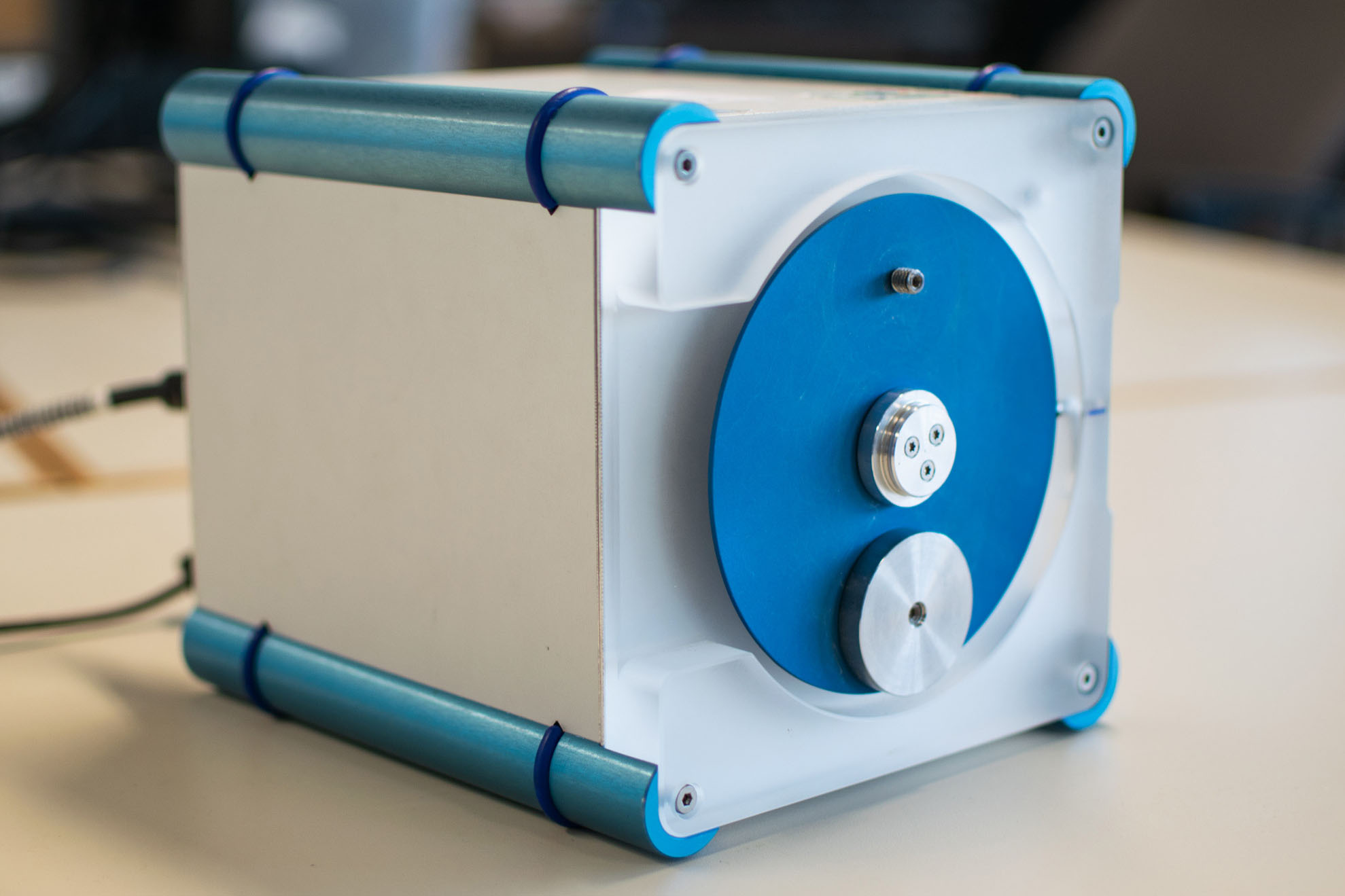}
	\caption{Unbalanced disk setup.}\label{4_fig:ubdisk}
	\vspace{-1em}
\end{figure}	
\extver{In the next example, we}{We}demonstrate the US controller design on an unbalanced disk experimental setup, see \cref{4_fig:ubdisk}. \extver{Like for the Duffing oscillator example, we}{We}consider a constant reference tracking disturbance rejection problem and \extver{also}{}compare the achieved performance to a standard LPV controller design. By neglecting the fast electrical dynamics of the motor, the motion of the unbalanced disk can be described as
\begin{subequations}\label{4_eq:disc}
\begin{align}
    \dot{\theta}(t) &= \omega(t);\\
    \dot{\omega}(t) &= \tfrac{M g l}{J}\sin(\theta(t)) -\tfrac{1}{\tau}\omega(t)+\tfrac{K_m}{\tau}V(t);
\end{align}
\end{subequations}
where $\theta$ $[\mr{rad}]$ is the angle of the disk, $\omega$ $[\mr{rad\cdot s^{-1}}]$ its angular velocity, $V$ $[\mr{V}]$ is the input voltage to the motor, $g$ is the gravitational acceleration, $l$ the length of the pendulum, $J$ the inertia of the disk, and $K_m$ and $\tau$ are the motor constant and time constant respectively. The angle of the disk $\theta$ is considered to be the output of the plant.
The physical parameters are given in \cref{4_table:disc}.
\begin{table}
\centering
\caption{Physical parameters of the unbalanced disk.}
\label{4_table:disc}
\begin{tabular}{*7l}\toprule
{Parameter:} & $g$ & $J$ & $K_m$ & $l$ & $M$ & $\tau$ \\
Value: & $9.8$ & $2.4\e{-4}$ & $11$ & $0.041$ & $0.076$ & $0.40$\\ \bottomrule
\end{tabular}
\end{table}

\begin{figure}
    \centering
    \includegraphics[scale=1]{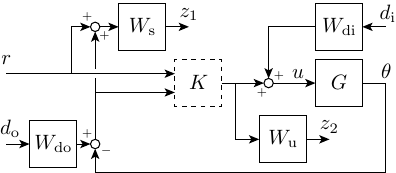}
    \caption{Generalized plant  considered for the unbalanced disk.}
    \label{4_fig:disgenplant3}
    \vspace{-1em}
\end{figure}

A generalized plant structure is used as depicted in \cref{4_fig:disgenplant3}, where $\plant$ is the system given by \cref{4_eq:disc}, $\controller$ is the to-be-designed controller, $\gd =\col\left( \rf ,\dist_\mathrm{i}, \dist_\mathrm{o}\right)$ is the generalized disturbance, with $\rf$ the reference, $\dist_\mathrm{i}$ being an input disturbance, and $\dist_\mr{o}$ an output disturbance. In this case, the controller $\controller$ has a two-degree-of-freedom structure, meaning the tracking error and reference trajectory are separate inputs to the controller. The weighting filters are chosen as
\begin{equation}
\begin{alignedat}{2}
    &W_\mathrm{s}(s)= \frac{0.5012 s + 2.005}{s+0.02005}, \qquad &&W_\mathrm{u}(s) = \frac{s+40}{s+4000},\\
    &W_\mathrm{di} = 0.5, \qquad &&W_\mathrm{do} = 0.1.
    \end{alignedat}
\end{equation}
Note that the integral action is approximate in this case, due to the choice of $W_\mr{s}$, as $W_\mr{s}$ includes a real pole close to the origin. While this means that we are not able to use \cref{4_cor:intrealiz} for the realization step of the US controller design, we will see that this choice will still result in good tracking and rejection behavior of the closed-loop.

As the unbalanced disk \cref{4_eq:disc} is the only nonlinear system in the generalized plant, only a VPV embedding of \cref{4_eq:disc} has to be constructed, which we take as:
\begin{equation}\label{4_eq:discincrLPV}
\begin{aligned}
    \dot\veloform{\theta}(t) &= \veloform{\omega}(t);\\
    \dot\veloform{\omega}(t) &= \left(\tfrac{M g l}{J}\sch(t))\right)\veloform{\theta}(t) -\tfrac{1}{\tau}\veloform{\omega}(t)+\tfrac{K_m}{\tau}\veloform{V}(t);
\end{aligned}
\end{equation}
where $\sch(t)=\schMap(\theta(t))=\cos(\theta(t))$ is the scheduling-variable which is assumed to be in $\schSet=[-1,\, 1]$, resulting in the VPV embedding region $\stSetLPV = \reals\times\reals$ for \cref{4_eq:disc}. Note that $\dot{\sch}(t) = -\sin(\theta(t))\omega(t)$ for which no bounds are explicitly assumed.

\extver{Like in the previous example, we}{We}use the LPVcore Toolbox\extver{}{\cite{DenBoef2021}}to synthesize the velocity controller, resulting in an \ltwo-gain of $\perf = 0.56$ for the velocity form of the closed-loop. This yields a US controller that ensures the closed-loop is USAS and has an \lstwo-gain of $\perf \leq 0.56$ (under \cref{4_prop:veloshifteddissip}).

For the unbalanced disk, a standard LPV controller is also designed for comparison. For this, the primal form of the nonlinear system  \cref{4_eq:disc} embedded in an LPV representation, which results in
\begin{equation}\label{4_eq:discLPV}
\begin{aligned}
    \dot{\theta}(t) &= \omega(t);\\
    \dot{\omega}(t) &= \left(\tfrac{M g l}{J}\sch_\mathrm{o}(t))\right) \theta(t) -\tfrac{1}{\tau}\omega(t)+\tfrac{K_m}{\tau}V(t);
\end{aligned}
\end{equation}
where $ \sch_\mathrm{o}(t) =\schMap_\mathrm{o}(\theta(t)) = \tfrac{\sin(\theta(t))}{\theta(t)} = \mr{sinc}(\theta(t))$. 
$\schSet_\mathrm{o}$ is chosen\footnote{Note that $\schMap_\mr{o}(0) = 1$ as $\lim_{x\rightarrow 0}\mr{sinc}(x)=1$.} as $[-0.22,\, 1]$. \extver{}{Here, we will denote with subscript `o' the standard concept of LPV embedding and control design.}The same generalized plant structure as for the US design is considered, see \cref{4_fig:disgenplant3}, and we also use the LPVcore Toolbox to synthesize the standard LPV controller. This then results in \ltwo-gain for the standard LPV controller design of $\perf = 0.56$.

\extver{The bode magnitude plot of the closed-loop with standard LPV controller and of the velocity form of the closed-loop with the US controller for frozen values of the scheduling-variables are plotted in \cref{4_fig:disc_bode}. From the figure, it can be seen the closed-loop frequency response behavior of the standard LPV controller and the velocity form of the US controller is almost identical. However, as we will see next, the closed-loop time-response behavior of the system with \emph{realized} US controller and the standard LPV controller is different.

\begin{figure}
    \centering
    \includegraphics[scale=1]{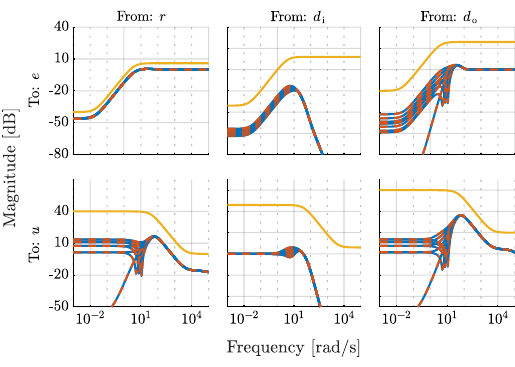}
    \caption{Bode magnitude plot of the closed-loop with standard LPV controller (\legendline{mblue}) and of the velocity form of the closed-loop with the US controller \mbox{(\legendline{morange,dashed})} for frozen values of the scheduling-variables, along with the corresponding inverse weightings \mbox{(\legendline{myellow})}.}    
   	\label{4_fig:disc_bode}
   	\vspace{-1em}
\end{figure}}{}

\begin{figure}
    \centering
    \includegraphics[scale=1]{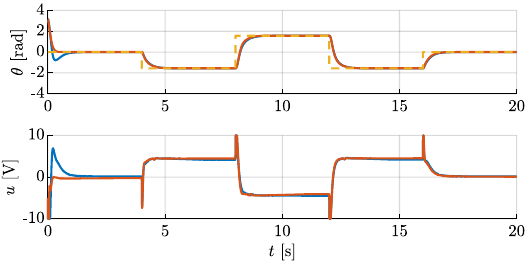}
    \caption{Measured angle of the unbalanced disk system (top)  in closed-loop with the standard LPV (\legendline{mblue}) and the US \mbox{(\legendline{morange})} controllers under reference \mbox{(\legendline{myellow,dashed})} and no input disturbance, together with inputs to the plant (bottom) generated by the controllers.}    
   	\label{4_fig:disc_exp_nodist}
   	\vspace{-1em}
\end{figure}
\begin{figure}
    \centering
    \includegraphics[scale=1]{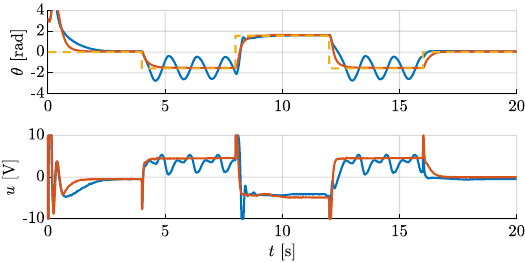}
    \caption{Measured angle of the unbalanced disk system (top) in closed-loop with the standard LPV  (\legendline{mblue}) and the US \mbox{(\legendline{morange})} controller under reference \mbox{(\legendline{myellow,dashed})} and input disturbance, together with corresponding inputs to the plant (bottom) generated by the controllers.}
    \label{4_fig:disc_exp_dist}
    \vspace{-1em}
\end{figure}

Both the US controller and LPV controller are then implemented on the experimental setup, whereby, for safety, the input voltage to the system was saturated between $\pm$ 10 V. In \cref{4_fig:disc_exp_nodist}, the angular position of the disk during the experiment is depicted along with the input to the plant (i.e., $u$) for a piecewise constant reference signal. Note that for the experiment, the disk starts in the downward position, which is why the initial angle is at $\pi$ radians.
In \cref{4_fig:disc_exp_dist}, the same reference trajectory is used, but a constant input disturbance of $d_\mathrm{i} = 60\,\textrm{V}$ is introduced (which is implemented by adding 60 V to the control input that is sent to the plant before saturation). For this input disturbance, the standard LPV controller performs much worse compared to the US controller design, which has similar performance to the case when no input disturbance is applied. Both the LPV controller and US controller are able to compensate the 60 V input disturbance, as visible in the control input that is sent to the plant (i.e., $\ipp$ in \cref{4_fig:disgenplant3}), see bottom graph in \cref{4_fig:disc_exp_dist}. However, while the control input that is sent to the plant is nearly identical for the US controller in both cases, see \cref{4_fig:disc_exp_nodist,4_fig:disc_exp_dist}, this is clearly not the case for the LPV controller, as oscillations in the signal are present when the input disturbance is applied, which causes unwanted oscillation of the disk angle. 
See also \url{https://youtu.be/B5453HP0APo} for a video of the experiment.

\section{Conclusions}\label{4_sec:conclusion}
In this paper, we have developed a novel systematic framework for the analysis and control of nonlinear systems to guarantee equilibrium free stability and performance. By using the LPV analysis and controller synthesis and by using the time-differentiated dynamics of a system, we can give a computationally attractive analysis and controller design approach. The resulting controller design is especially beneficial in applications considering (constant) reference tracking and disturbance rejection, which has also been demonstrated through \extver{a simulation and}{an}experimental study. For future research, we aim to extend the results to discrete-time nonlinear systems.

\appendices
\section{Proofs}
\extver{\proofsection{4_thm:shiftlyapstab}\label[appendix]{4_pf:shiftlyapstab}
For every $(\stEq,\gdEq)\in\proj_\mr{\stEq,\gdEq}\eqSet$, the function $\lyapfun:\st\mapsto\lyapfunShift(\st,\gdEq)$ satisfies the conditions for a Lyapunov function for the equilibrium point $\stEq$, as $\lyapfun = (\st\mapsto\lyapfunShift(\st,\gdEq))\in\posClass{\stEq}$ for every $(\stEq,\gdEq)\in\proj_\mr{\stEq,\gdEq}\eqSet$. Consequently, by \cref{4_eq:shiftedstability}, it holds for every $(\stEq,\gdEq)\in\proj_\mr{\stEq,\gdEq}\eqSet$ that
	\begin{equation}\label{4_eq:lyaptoshift}
		\frac{d}{dt}\lyapfun(\st(t))\leq 0,
	\end{equation}
	for all $t\in\nnreals$ and $\st\in\proj_\mr{\st}\Bw(\gd\equiv\gdEq)$. Hence, by Lyapunov theory, see e.g. \cite{Khalil2002}, the system is stable at each equilibrium point $(\stEq,\gdEq,\gpEq)\in\eqSet$, meaning, by definition, it is USS. Similarly, when \cref{4_eq:shiftedstability} holds, but with a strict inequality except when $\st(t)=\stEq$, this implies that \cref{4_eq:lyaptoshift} holds, but with a strict inequality except when $\st(t)=\stEq$. Therefore, from Lyapunov stability theory \cite{Khalil2002}, asymptotic stability of the nonlinear system follows at each equilibrium point $(\stEq,\gdEq,\gpEq)\in\eqSet$, meaning the system is USAS.}{}

\extver{\proofsection{4_lem:ls2gaindissip}\label[appendix]{4_pf:ls2gaindissip}
If the system is $\qsr$-USD with $\qsr = (\perf^2 I,0,-I)$, it holds that there exists a $\storfunShift$ such that for every $(\stEq,\gdEq,\gpEq)\in\eqSet$
\begin{multline}
	\storfunShift(\st(t_1),\gdEq)-\storfunShift(\st(t_0),\gdEq)\leq\\ \int_{t_0}^{t_1} \perf^2 (\star)^\top (\gd(t)-\gdEq) -  (\star)^\top (\gp(t)-\gpEq)\,dt,
\end{multline}
for all $t_0,t_1\in\nnreals$ with $t_1\geq t_0$ and $(\st,\gd,\gp)\in\B$. Consequently, it also holds for every $(\stEq,\gdEq,\gpEq)\in\eqSet$
 that 
\begin{multline}
	0\leq \storfunShift(\st(t),\gdEq)\leq\\ \int_{0}^{T} \perf^2 \norm{\gd(t)-\gdEq} -  \norm{\gp(t)-\gpEq}\,dt + \storfunShift(\stIc,\gdEq),
\end{multline}
for all $T\geq 0$ and $(\st,\gd,\gp)\in\B$. Using that $\gdEq=\eqMap(\stEq)$, this is equivalent to
\begin{equation}\label{4_eq:l2expressioninproof2}
\norm{\gp-\gdEq}_{2,T}^2 \le \perf^2 \norm{\gd-\gpEq}_{2,T}^2 + \storfunShift(\stIc,\eqMap(\stEq)),
\end{equation}
holding for all $T\ge 0$ and $(\st,\gd,\gp)\in\B$ with $\gd\in\ltwoe$. Next, we take the square root on both sides of \cref{4_eq:l2expressioninproof2}, which gives us
\begin{multline}
\norm{\gp-\gpo}_{2,T} \le \sqrt{\perf^2 \norm{\gd-\gdo}_{2,T}^2 + \storfunIncr(\stIc,\sto(0))}\leq \\\perf\norm{\gd-\gdo}_{2,T}+\sqrt{\storfunShift(\stIc,\eqMap(\stEq))},
\end{multline}
 which is equivalent to \cref{4_eq:shiftlplqgain} with $\icfunShift(\stIc,\stEq) = \sqrt{\storfunShift(\stIc,\eqMap(\stEq))}$.}{}
 
\extver{\proofsection{4_thm:shiftdissipstab}\label[appendix]{4_pf:shiftdissipstab}
If the system given by \cref{4_eq:nonlinsys} is USD w.r.t. a supply function $\supfunShift$, for which 
 \cref{4_eq:supfunshiftcond} holds for all $\gp\in\gpSet$ and every $(\stEq,\gdEq,\gpEq)\in\eqSet$, then \begin{multline}\label{4_eq:storfunshiftproof}
\storfunShift(\st(t_1),\gdEq)-\storfunShift(\st(t_0),\gdEq)\leq\\ \int_{t_0}^{t_1} \supfunShift(\gdEq,\gdEq,\gp(t),\gpEq)\,dt\leq 0,
\end{multline}
for all $t_0,t_1\in\nnreals$ with $t_0 \le t_1$, $(\st,\gp)\in\proj_\mr{\st,\gp} \Bw(\gd\equiv\gdEq)$, and for every $(\stEq,\gdEq,\gpEq)\in\eqSet$. As $\storfunShift(\cdot,\gdEq)\in\C{1}$ for all $\gdEq\in\gdSetEq$, this means that for every $(\stEq,\gdEq,\gpEq)\in\eqSet$ it holds that 
 \begin{equation}\label{4_eq:storfunshiftproof2}
 \frac{d}{dt}\storfunShift(\st(t),\gdEq)\leq 0,
 \end{equation}
 for all $t\in\nnreals$ and $\st\in\proj_\mr{\st} \Bw(\gd\equiv\gdEq)$. The storage function $\storfunShift$ satisfies the conditions for the universal shifted Lyapunov function $\lyapfunShift$ in \cref{4_thm:shiftlyapstab}. Hence, \cref{4_eq:storfunshiftproof2} implies \cref{4_eq:shiftedstability}, which by \cref{4_thm:shiftlyapstab} implies USS.

In case of USAS, the supply function satisfies \cref{4_eq:supfunshiftcond}, but with strict inequality for all $\gp\neq \gpEq$. Moreover, the system is assumed to be observable. Therefore, $\gp\in\proj_\mr{\gp}\Bw(\gd\equiv\gdEq)$ for which $\gp(t)=\gpEq\in\proj_{\mr{\gpEq}}\eqSet$ for all $t\in\nnreals$ implies that $\st(t)=\stEq\in\proj_{\mr{\stEq}}\eqSet$, where $(\stEq,\gdEq,\gpEq)\in\eqSet$. Consequently, we have that \cref{4_eq:storfunshiftproof2} holds, but with strict inequality except when $\st(t)=\stEq$, which by \cref{4_thm:shiftlyapstab} implies USAS.}{}

\extver{\proofsection{4_thm:veloqsrMI}\label[appendix]{4_pf:veloqsrMI}
If \cref{4_eq:veloMI} holds for all $(\st,\gd)\in\stSet\times\gdSet$, then, by pre- and post multiplication of \cref{4_eq:veloMI} with $\col(\stMap(\st,\gd),\gddot)^\top$ and $\col(\stMap(\st,\gd),\gddot)$, we get
\begin{multline}\label{4_eq:thmveloqsrmi}
	2(\stMap(\st,\gd)^\top \storquad (\stMap(\st,\gd))\left(\velA(\st,\gd)\stMap(\st,\gd) +\velB(\st,\gd)\gddot\right)-\\ \gddot^\top \supQ \gddot - 2 \gddot^\top \supS \big(\velC(\st,\gd)\stMap(\st,\gd) + \velD(\st,\gd)\gddot\big)-\\
(\star)^\top  \supR \big(\velC(\st,\gd)\stMap(\st,\gd) + \velD(\st,\gd)\gddot\big) \leq 0,
\end{multline}
all $\gddot\in\reals^\gdSize$ and $(\st,\gd)\in\stSet\times\gdSet$. In fact, \cref{4_eq:thmveloqsrmi} is
\begin{equation}\label{4_eq:velodissipproof}
	\frac{d}{dt}\storfunVelo(\dot{\st}(t))\leq \supfunVelo(\dot{\gd}(t),\dot{\gp}(t)),
\end{equation}
where $\storfunVelo$ is given by \cref{4_eq:velostorquad} and $\supfunVelo$ by \cref{4_eq:velosupply}. By integrating \cref{4_eq:velodissipproof} from $t_0$ to $t_1$ with $t_1\geq t_0\geq 0$, we obtain \cref{4_eq:velodissip}. Hence, the system given by \cref{4_eq:nonlinsys} is $\qsr$-VD.}{}

\proofsection{4_thm:velotoshiftstab}\label[appendix]{4_pf:velotoshiftstab}
For each equilibrium point $(\stEq,\gdEq,\gpEq)\in\eqSet$, consider
\begin{equation}
	\lyapfunShift(\st(t),\gdEq)=\lyapfunVelo(\stMap(\st(t),\gdEq))=\lyapfunVelo(\dot\st(t)). 
\end{equation}
This choice implies that $\storfunShift(\cdot,\gdEq)\in\posClass{\stEq}$ and $\storfunShift(\cdot,\gdEq)\in\C{1}$ as $\storfunVelo\in\posClass{0}$, $\storfunVelo\in\C{1}$, and $\stMap\in\C{1}$. Note that this requires uniqueness of the equilibrium points (see \cref{4_assum:uniqueEq}), as otherwise there exists multiple $\stEq$ for which $\lyapfunShift(\stEq,\gdEq)=0$. By this choice of $\lyapfunShift$, we have for each $(\stEq,\gdEq,\gpEq)\in\eqSet$, that
\begin{equation}\label{4_eq:shiftvelostab}
	\frac{d}{dt}\lyapfunShift(\st(t),\gdEq) = \frac{d}{dt}\lyapfunVelo(\dot{\st}(t))\leq 0,
\end{equation}
for all $t\in\nnreals$ and $\dot\st\in\proj_\mr{\dot{\st}}\Bvw(\gd\equiv \gdEq)$ and correspondingly for all $\st\in\proj_\mr{\st}\Bcw(\gd\equiv \gdEq)$. This implies that \cref{4_eq:shiftedstability} holds for all $\st\in\proj_\mr{\st}\Bcw(\gd\equiv\gdEq)$ and for all equilibrium points $(\stEq,\gdEq)\in\proj_\mr{\stEq,\gdEq}\eqSet$. Hence, by \cref{4_thm:shiftlyapstab}, the system is USS. The asymptotic stability version follows similarly by changing \cref{4_eq:shiftvelostab} to a strict inequality.

\proofsection{4_lem:velostab}\label[appendix]{4_pf:velostab}
If the system given by \cref{4_eq:nonlinsys} is VD w.r.t. a supply function $\supfunVelo$ which satisfies \cref{4_eq:supplystability} for all $\gpdot\in\reals^\gpSize$, then it holds that
\begin{equation}\label{4_eq:velodissippf1}
	\storfunVelo(\dot{\st}(t_1))-\storfunVelo(\dot{\st}(t_0))\leq \int_{t_0}^{t_1} \supfunVelo(0,\dot{\gp}(t))\,dt\leq 0,
\end{equation}
for all $t_0,t_1\in\nnreals$ with $t_1\geq t_0$ and $(\dot{\st},\dot{\gp})\in\proj_\mr{\dot\st,\dot\gp}\Bvset{\gdSetEq}$. As $\storfunVelo\in\C{1}$, this implies that 
\begin{equation}\label{4_eq:velodissippf2}
	\frac{d}{dt}\storfunVelo(\dot{\st}(t))\leq 0,
\end{equation}
for all $t\in\nnreals$ and $\dot{\st}\in\proj_\mr{\dot\st}\Bvset{\gdSetEq}$, which yields universal shifted stability through \cref{4_thm:velotoshiftstab}, with $\storfunVelo = \lyapfunVelo$.

For universal shifted asymptotic stability, $\supfunVelo$ satisfies \cref{4_eq:supplystability}, but with a strict inequality when $\gpdot\neq 0$. Note that $\dot\gp\in\Bv$ can be associated with any $\gp\in\Bc$ such that $\frac{d}{dt}\gp(t)=\dot\gp(t)$. As it is assumed that the system is observable, this means that $\dot\gp(t)=0$ for all $t\in\nnreals$ (corresponding to $\gp\in\proj_\mr{\gp}\Bcw(\gd\equiv\gdEq)$ for which $\gp(t)=\gpEq\in\proj_{\mr{\gpEq}}\eqSet$ for all $t\in\nnreals$) implies that $\st(t)=\stEq$, where $(\stEq,\gdEq,\gpEq)\in\eqSet$. Consequently, \cref{4_eq:velodissippf2} is satisfied with a strict inequality except when $\dot{\st}(t)=0$, which implies universal shifted asymptotic stability of the system according to \cref{4_thm:velotoshiftstab}.

\proofsection{4_thm:veloshiftperf}\label[appendix]{4_pf:veloshiftperf}
If the nonlinear system is VD w.r.t. the supply function $\supfunVelo(\dot\gd,\dot\gp) = \dot\gd ^\top \supQ \dot\gd+\dot\gp^\top \supR \dot\gp$, there exists a storage function $\storfunVelo$, such that for all $t_0,t_1\in\nnreals$ with $t_1\geq t_0$
\begin{equation}
	\storfunVelo(\dot\st(t_1))-\storfunVelo(\dot\st(t_0))\leq \int_{t_0}^{t_1} \dot\gd(t)\!^\top \supQ \dot\gd(t)+\dot\gp(t)\!^\top \supR \dot\gp(t)\,dt,
\end{equation}
for all $(\dot{\st},\dot{\gd},\dot{\gp})\in\Bv$, corresponding to $(\st,\gd,\gp)\in\Bc$. Hence, as $\storfunVelo(\dot \st(0))=\storfunVelo(0)=0$ and $\storfunVelo(\stdot)> 0,\,\forall\,\stdot\in\reals^{\stSize}\backslash\{0\}$ this implies that
\begin{equation}\label{4_eq:pf:qrvsp1}
	0\leq \int_{0}^{T} \dot\gd(t)\!^\top \supQ \dot\gd(t)+\dot\gp(t)\!^\top \supR \dot\gp(t)\,dt,
\end{equation}
for all $T>0$ and $(\dot{\st},\dot{\gd},\dot{\gp})\in\Bv$. 
Defining $\tilde \supQ=\frac{1}{\norm{\supQ}}\supQ$ and $\tilde \supR=\frac{1}{\norm{\supQ}}\supR$, it also holds that 
\begin{equation}\label{4_eq:pf:qrvsp1.1}
	0\leq \int_{0}^{T} \dot\gd(t)\!^\top \tilde \supQ \dot\gd(t)+\dot\gp(t)\!^\top \tilde \supR \dot\gp(t)\,dt,
\end{equation}
Next, using that $\stMap(\stEq)+\ltiB \gdEq = 0,\forall(\stEq,\gdEq,\gpEq)\in\eqSet$, and $\dot{\gp}(t) = \ltiC\dot{\st}(t)$, we have 
\begin{subequations}
\begin{align}
	\dot{z}^\top \tilde \supR \dot\gp &= \dot\st^\top \ltiC^\top \tilde \supR\, \ltiC\dot\st,\\
	&= (\star)\!^\top \tilde \supR\, \ltiC(\stMap(\st)+\ltiB \gd),\\
	&= (\star)\!^\top \tilde \supR\,\ltiC(\stMap(\st)+\ltiB \gd\underbrace{-(\stMap(\stEq)+\ltiB \gdEq)}_{=0}),\\
	&= (\star)\!^\top \tilde \supR\,\ltiC(\stMap(\st)-\stMap(\stEq)+\ltiB (\gd -\gdEq)).\label{4_eq:perfproofeq1}
\end{align}
\end{subequations}
Through, the fundamental theorem of calculus we have 
\begin{equation}
\begin{aligned}
	\stMap(\st)-\stMap(\stEq)&=\left(\int_0^1\Partial{\stMap}{\st}(\stEq+\var(\st-\stEq))\,d\var\right)(\st-\stEq), \\
	&= \underbrace{\left(\int_0^1 \velA(\stEq+\var(\st-\stEq))\,d\var\right)}_{\intA(\st,\stEq)}(\st-\stEq).
\end{aligned}
\end{equation}
Combining this with \cref{4_as:CB}, we can write \cref{4_eq:perfproofeq1} as
\begin{equation}
	\dot{z}^\top \tilde \supR \dot\gp = (\star)\!^\top \tilde \supR\,\ltiC\intA(\st,\stEq)(\st-\stEq).
\end{equation}
Next, by satisfying \cref{4_as:veloShiftBound2} for $T=\tilde\supR\preceq 0$, we have
\begin{multline}\label{4_eq:pf:qrvsp2}
	\dot\gp^\top \tilde \supR \dot\gp= (\star)^\top \tilde \supR\,\ltiC\intA(\st,\stEq)(\st-\stEq) \leq\\
	\alpha^{-1}(\star)^\top \tilde \supR\, \ltiC(\st-\stEq)= \alpha^{-1}(\star)^\top \tilde \supR(\gp-\gpEq),
\end{multline}
for every $\stEq\in\stSetEq$. Moreover, by \cref{4_as:wExoSys},
\begin{equation}\label{4_eq:pf:qrvsp3}
	\dot\gd^\top \tilde \supQ \dot\gd=(\star)\!^\top \tilde \supQ \exoA(\gd-\gdEq)\leq \beta^2 (\star)\!^\top \tilde \supQ (\gd-\gdEq),
\end{equation}
for a given $(\stEq,\gdEq,\gpEq)\in\eqSet$, where $\gd\in\exoBvr_{\gdEq}$ and $0\preceq\tilde \supQ \preceq I$.
Combining \cref{4_eq:pf:qrvsp1.1,4_eq:pf:qrvsp2,4_eq:pf:qrvsp3},
\begin{equation}
	\int_{0}^{T} \beta^2(\star)\!^\top \tilde \supQ (\gd(t)-\gdEq)+\alpha^{-1}(\star)\!^\top \tilde \supR (\gp(t)-\gpEq)\,dt\geq 0,
\end{equation}
holds for every $(\stEq,\gdEq,\gpEq)\in\eqSet$ and for all $T> 0$ and $(\gd,\gp)\in\proj_\mr{\gd,\gp}\Bc$ with $\gd\in\exoBvr_{\gdEq}$.
Hence, also 
\begin{equation}
	\int_{0}^{T} \beta^2(\star)\!^\top \supQ (\gd(t)-\gdEq)+\alpha^{-1}(\star)\!^\top  \supR (\gp(t)-\gpEq)\,dt \geq0,
\end{equation}
for all $T> 0$ and $(\gd,\gp)\in\proj_\mr{\gd,\gp}\Bc$ with $\gd\in\exoBvr_{\gdEq}$. 

\proofsection{4_cor:veloshiftl2}\label[appendix]{4_pf:veloshiftl2}
Based on \cref{4_thm:veloshiftperf} with $\qsr=(\perf^2I,0,-I)$, there exists a function $\icfunShift :\stSet\times\stSetEq \to\reals$, such that	\begin{multline}
	\int_{0}^{T} \perf^2\beta^2(\star)^\top (\gd(t)-\gdEq)-\alpha^{-1}(\star)^\top (\gp(t)-\gpEq)\,dt\geq0,
\end{multline}
for all $T>0$, $(\gd,\gp)\in\proj_\mr{\gd,\gp}\Bc$ with $\gd\in\exoBvr$, and for every $(\stEq,\gdEq,\gpEq)\in\eqSet$. This is equivalent to
\begin{equation}
\alpha^{-1} \int_{0}^{T}\norm{\gp(t)-\gpEq}^2 \,dt \leq \perf^2\beta^2 \int_{0}^{T}\norm{\gd(t)-\gdEq}^2\,dt
\end{equation}
\begin{equation}
\norm{\gp-\gpEq}_{2,T}^2 \leq \alpha\perf^2\beta^2 \norm{\gd-\gdEq}_{2,T}^2,
\end{equation}
for all $T>0$ and $(\gd,\gp)\in\proj_\mr{\gd,\gp}\Bc$ with $\gd\in\exoBvr$. Hence, this implies that for every $(\stEq,\gdEq,\gpEq)\in\eqSet$
\begin{equation}
	\norm{\gp-\gpEq}_{2,T} \leq \tilde\perf \norm{\gd-\gdEq}_{2,T},
\end{equation}
for all $T>0$ and $(\gd,\gp)\in\proj_\mr{\gd,\gp}\Bc$ with $\gd\in\exoBvr$ where $\tilde\perf = \alpha\perf^2\beta^2$.

\proofsection{4_lem:vpvembed}\label[appendix]{4_pf:vpvembed}
The LPV representation \cref{4_eq:vpv} is a VPV embedding of the system given by \cref{4_eq:nonlinsys} on the region $\stSetLPV\times\gdSetLPV= \stSet\times\gdSet$. Consequently, for any trajectory $(\dot\st,\dot\gd,\dot\gp)\in\Bv$, we also have that $(\dot\st,\dot\gd,\dot\gp)\in\Blpv[\schMap(\st,\gd)]$. Moreover, as $\stSetLPV\times\gdSetLPV= \stSet\times\gdSet$ and $\schMap(\stSetLPV,\gdSetLPV)\subseteq\schSet$, we obtain the following relation
\begin{multline}
\Bv = \bigcup_{(\st,\gd)\in\proj_{\mr{\st,\gd}}\B} \Blpv[\schMap(\st,\gd)]\subseteq\\ \bigcup_{(\st,\gd)\in(\stSet,\gdSet)^\nnreals} \Blpv[\schMap(\st,\gd)]\subseteq \bigcup_{\sch\in\schSet^\nnreals} \Blpv = \Blpvfull.
\end{multline}

\proofsection{4_thm:velodissiplpv}\label[appendix]{4_pf:velodissiplpv}
As the LPV representation \cref{4_eq:vpv} is a VPV embedding of the nonlinear system given by \cref{4_eq:nonlinsys} on the region $\stSetLPV\times\gdSetLPV= \stSet\times\gdSet$, we have by \cref{4_lem:vpvembed} that the LPV representation describes the full behavior of the velocity form \cref{4_eq:veloform}, i.e., $\Bv \subseteq \Blpvfull$. Consequently, if the LPV representation \cref{4_eq:vpv} is classically dissipative for all trajectories in $\Blpvfull$, we have that the velocity form is classically dissipative for all trajectories in $\Bv$, which corresponds to the nonlinear system being VD, see \cref{4_def:velodissip}.

\proofsection{4_thrm:veloICL2}\label[appendix]{4_pf:veloICL2}
As $\velogenplantLPV$ is a VPV embedding on the region $\stpSetLPV\times\ippSetLPV\subseteq\stpSet\times\ippSet$, we have through \cref{4_lem:vpvembed} that $\BvXU\subseteq\Blpvfull$. Consequently, through \cref{4_thm:velodissiplpv}, we have that $\ic{\velogenplant}{\velocontroller}$ with $\sch=\schMap(\stp,\ipp)$ for $\velocontroller$ is classically dissipative and has an \ltwo-gain bound 
	$\leq\perf$ for all $(\dot\stp,\dot\ipp)\in\proj_\mr{\dot\stp,\dot\ipp}\BvXU$.
	
\proofsection{4_thrm:timediff}\label[appendix]{4_pf:timediff}
We introduce the operators $\dt = \frac{d}{dt}$ and $\dt^{-1}\gd(t) = \int_{0}^t \gd(\tau)\,d\tau$. To integrate the inputs and differentiate the outputs of the nonlinear system given by \cref{4_eq:nonlinsys}, we introduce new input and output $\hat{\gd}$ and $\hat{\gp}$, respectively, with the following relations:
\extver{\vspace{-.5em}}{}\begin{equation}\label{4_eq:pfdifint}
    \gd(t) = \dt^{-1} \hat{\gd}(t), \qquad    \hat{\gp}(t) = \dt \gp(t),
\end{equation}
which exist as the solutions $(\st,\gd,\gp)\in\Bc$. Obviously, multiplication with $\dt$ or $\dt^{-1}$ is not commutative. We can then write \cref{4_eq:nonlinsys} as
\begin{subequations}
\begin{align}
    \dt \st(t) &= \stMap(\st(t),\dt^{-1} \hat{\gd}(t));\\
    \dt^{-1} \hat{\gp}(t) &= \opMap(\st(t),\dt^{-1} \hat{\gd}(t));
   \end{align}
\end{subequations}
which, by multiplication with $\dt$ to obtain the new output $\hat{\gp}$, results in
\begin{subequations}
\begin{align}
    \dt^2 \st(t) &= \dt \stMap(\st(t),\dt^{-1} \hat{\gd}(t));\\
    \hat{\gp}(t) &= \dt \opMap(\st(t),\dt^{-1} \hat{\gd}(t));
\end{align}
\end{subequations}
\begin{subequations}
\begin{align}
    \dt^2 \st(t) &= \frac{\partial \stMap(\st(t),\dt^{-1} \hat{\gd}(t))}{\partial x}\dt \st(t)+\notag\\&\phantom{bladie}\frac{\partial \stMap(\st(t),\dt^{-1} \hat{\gd}(t))}{\partial \gd}\dt \left(\dt^{-1}\hat{\gd}(t)\right);\\
    \hat{\gp}(t) &= \frac{\partial \opMap(\st(t),\dt^{-1}\hat{\gd}(t))}{\partial x}\dt \st(t)+\notag\\&\phantom{bladie}\frac{\partial \opMap(\st(t),\dt^{-1} \hat{\gd}(t))}{\partial \gd}\dt \left(\dt^{-1}\hat{\gd}(t)\right);
\end{align}
\end{subequations}
\begin{subequations}\label{4_eq:pflasteq}
\begin{align}
    \dt^2 \st(t) &= \frac{\partial \stMap(\st(t),\gd(t))}{\partial x}\dt \st(t)+\frac{\partial \stMap(\st(t),\gd(t))}{\partial \gd} \hat{\gd}(t);\\
    \hat{\gp}(t) &= \frac{\partial \opMap(\st(t),\gd(t))}{\partial x}\dt \st(t)+\frac{\partial \opMap(\st(t),\gd(t))}{\partial \gd} \hat{\gd}(t).
\end{align}
\end{subequations}
By using the definitions in \cref{4_eq:pfdifint}, we can express \cref{4_eq:pflasteq} as
\begin{subequations}\label{eq:velosys}
\begin{align}
    \ddot\st(t) &= \frac{\partial \stMap(\st(t),\gd(t))}{\partial x}\dot{\st}(t)+\frac{\partial \stMap(\st(t),\gd(t))}{\partial \gd} \dot{\gd}(t);\\
    \dot{\gp}(t) &= \frac{\partial \opMap(\st(t),\gd(t))}{\partial x}\dot{\st}(t)+\frac{\partial \opMap(\st(t),\gd(t))}{\partial \gd} \dot{\gd}(t).
\end{align}
\end{subequations}
From this, it is clear that we obtain the velocity form of \cref{4_eq:nonlinsys} given by \cref{4_eq:veloform}.

\vspace{-.5em}
\extver{\vspace{-.5em}}{}
\proofsection{4_thrm:controlrealiz}\label[appendix]{4_pf:controlrealiz}
First, the primal form  of the controller is realized. To realize the primal form we differentiate the input to the velocity controller $\ipkdot$ and integrate the output $\opkdot$ of the velocity controller $\velocontroller$ given by \cref{4_eq:veloContr}. Therefore, the following relations hold\vspace{-.5em}
\begin{equation}\label{4_eq:diffint}
\dt \opk = \opkdot, \qquad \dt \ipk = \ipkdot,
\end{equation}
where again $\dt = \frac{d}{dt}$ and $\dt^{-1}\gd(t) = \int_{0}^t \gd(\tau)\,d\tau$.
By simply rewriting \cref{4_eq:veloContr}, we get 
\begin{subequations}\label{4_eq:deltaK}
\begin{align}
    \dt \stkdot &= \lpvAk(\sch)\stkdot+\lpvBk(\sch) \dt \ipk;\label{4_eq:deltaKstk}\\
    \dt \opk &= \lpvCk(\sch)\stkdot+\lpvDk(\sch) \dt \ipk\label{4_eq:deltaKopk};
\end{align}
\end{subequations}
Based on \cref{4_eq:deltaKstk}, we can perform the following transformations
\begin{equation}
\dt \stkdot \!=\! \lpvAk(\sch)\stkdot+ \lpvBk(\sch) \dt \ipk + (\dt \lpvBk(\sch))\ipk-(\dt \lpvBk(\sch))\ipk,
\end{equation}
\begin{equation}
\dt \stkdot = \lpvAk(\sch)\stkdot+ \dt\left(\lpvBk(\sch)\ipk\right)-(\dt \lpvBk(\sch))\ipk,
\end{equation}
\begin{equation}
\dt \left(\stkdot-\lpvBk(\sch)\ipk\right) = \lpvAk(\sch)\stkdot - (\dt \lpvBk(\sch))\ipk, \label{4_eq:stkstuffpf}
\end{equation}
we then define $\stkA = \stkdot - \lpvBk(\sch)\ipk$, resulting in
\begin{equation}
    \dt \stkA = \lpvAk(\sch)\stkA+\lpvAk(\sch)\lpvBk(\sch)\ipk-(\dt \lpvBk(\sch))\ipk, 
\end{equation}
which can be rewritten to
\begin{equation}
    \dotstkA = \lpvAk(\sch)\stkA+\left(\lpvAk(\sch)\lpvBk(\sch)-\dif\lpvBk(\sch,\dot\sch)\right)\ipk, 
\end{equation}
where $\dif\lpvBk(\sch,\dot\sch) = \sum_{i=1}^\schSize\Partial{\lpvBk(\sch)}{\sch_i}\dot\sch_i$. 

A similar procedure can be used to rewrite \cref{4_eq:deltaKopk}, resulting in
\begin{equation}\label{4_eq:dotskb}
    \dotstkB = \lpvCk(\sch)\stkA+\left(\lpvCk(\sch)\lpvBk(\sch)-\dif\lpvDk(\sch,\dot\sch)\right)\ipk, 
\end{equation}
where $\stkB = \opk-\lpvDk(\sch)\ipk$ and $\dif\lpvDk(\sch,\dot\sch) = \sum_{i=1}^\schSize\Partial{\lpvDk(\sch)}{\sch_i}\dot\sch_i$. Due to the definition of $\stkB$, we then have that $\opk = \stkB+\lpvDk(\sch)\ipk$. Combining these results gives us the primal realization $\controller$:\begin{equation}\label{4_eq:controlller}
\left[\begin{NiceArray}{c}[margin]
   \dotstkA \\\dotstkB \\ \hdottedline \opk
    \end{NiceArray}\;\,\right]\!\!=\!\!
    \left[\begin{NiceArray}{cc:c}[margin]
    \lpvAk(\sch) & 0 & \lpvAk(\sch)\lpvBk(\sch)-\dif\lpvBk(\sch,\dot\sch)\\\lpvCk(\sch) & 0 &\lpvCk(\sch)\lpvBk(\sch)-\dif\lpvDk(\sch,\dot\sch)\\\hdottedline
    0 & I & \lpvDk(\sch)
    \end{NiceArray}\right]\!\!
    \left[\begin{NiceArray}{c}[margin]
    \stkA\\\stkB\\\hdottedline \ipk
    \end{NiceArray}\;\,\right]\!.
\end{equation}
From which \cref{4_eq:Contr} can be constructed by introducing $\stkus = \col(\stkA,\stkB)$ and using \cref{4_eq:shiftcontrmat}. 

Next, we show that the velocity form of \cref{4_eq:controlller} is equal to \cref{4_eq:veloContr}.
As $\ipkdot,\opkdot,\sch\in\C{1}$, we have by \cref{4_thrm:timediff} that the controller given by \cref{4_eq:Contr} with its inputs integrated and outputs differentiated is equal to its velocity form. 
Substituting the relations of \cref{4_eq:diffint} in the first row of \cref{4_eq:controlller} and rewriting it gives
\begin{equation}
	\dt \stkA = \lpvAk(\sch)\stkA+\lpvAk(\sch)\lpvBk(\sch)\dt^{-1}\ipkdot-\left(\dt \lpvBk(\sch)\right)\dt^{-1}\ipkdot,
\end{equation}
for which we can write
\begin{align}
&\begin{aligned}
\dt \stkA = &\lpvAk(\sch)\stkA+\lpvAk(\sch)\lpvBk(\sch)\dt^{-1}\ipkdot-\\&
\underbrace{\left(\dt \lpvBk(\sch)\right)\dt^{-1}\ipkdot-\lpvBk(\sch)\ipkdot}_{\dt\big(\lpvBk(\sch)\dt^{-1}\ipkdot\big)}
	+\lpvBk(\sch)\ipkdot,
\end{aligned}\end{align}\begin{align}
&\dt \big(\stkA+\lpvBk(\sch)\dt^{-1}\ipkdot\big)= \notag\\&\phantom{blablabla}\lpvAk(\sch)\left(\stkA+\lpvBk(\sch)\dt^{-1}\ipkdot\right)+\lpvBk(\sch)\ipkdot,\label{4_eq_pf:stk-rewr-rev}
\end{align}  
then, by defining $\stkdot = \stkA+\lpvBk(\sch)\dt^{-1}\ipkdot$, \cref{4_eq_pf:stk-rewr-rev} results in
\begin{align}\label{eq:diffKx}
	\dt \stkdot &= \lpvAk(\sch)\stkdot+\lpvBk(\sch)\ipkdot.
\end{align}
Similarly, based on the second row of \cref{4_eq:controlller}, we can find that
\begin{multline}\label{4_eq:KxhatRewr} 
\dt \big(\stkB+\lpvDk(\sch)\dt^{-1}\ipkdot\big)=\\\lpvCk(\sch)\left(\stkA+\lpvBk(\sch)\dt^{-1}\ipkdot\right)+\lpvDk(\sch)\ipkdot.
\end{multline}
Now, by using that $\stkB+\lpvDk(\sch)\dt^{-1}\ipkdot=\stkB+\lpvDk(\sch)u=\opk$ based on the third row of \cref{4_eq:controlller} we can employ $\stkdot = \stkA+\lpvBk(\sch)\dt^{-1}\ipkdot$ to rewrite \cref{4_eq:KxhatRewr} to
\begin{equation}\label{4_eq:diffKy}
	\dt \opk = \lpvCk(\sch)\stkdot+\lpvDk(\sch)\ipkdot. 
\end{equation}
Finally, by combining \cref{eq:diffKx} and \cref{4_eq:diffKy}, we arrive at the velocity form of the controller
\begin{subequations}\label{4_eq:diffcontrintrho}
	\begin{align}
		\dot\stkdot &= \lpvAk(\sch)\stkdot+\lpvBk(\sch)\ipkdot;\\
		\opkdot &= \lpvCk(\sch)\stkdot+\lpvDk(\sch)\ipkdot;
	\end{align}
\end{subequations}
which is equivalent with \cref{4_eq:veloContr}.

\vspace{-1em}\proofsection{4_thrm:clpls2}\label[appendix]{4_pf:clpls2}
For our generalized plant $\genplant$ given by \cref{4_eq:genplant} with behavior $\Bc$ and velocity form $\velogenplant$ given by \cref{4_eq:velogenplant}, we have by \cref{4_thrm:veloICL2} that $\velocontroller$ given in \cref{4_eq:veloContr} ensures classical dissipativity and a bounded \ltwo-gain of $\perf$ of the closed-loop $\ic{\velogenplant}{\velocontroller}$ on $\stpSetLPV\times\ippSetLPV$. Moreover, we consider the set $\tilde\gdSet\subseteq \gdSet$, for which $\stclSet = \stpSetLPV \times \stkSet $ is invariant, meaning  that for any  $\gd\in\tilde{\gdSet}^\nnreals$, the resulting $(\stp(t),\ipp(t))\in\stpSetLPV\times\ippSetLPV ,\, \forall\,t\in \nnreals$. Hence, we will remain in the design set on which classical dissipativity and a bounded \ltwo-gain of the velocity form are ensured. 

By \cref{4_thrm:controlrealiz}, we have that the velocity form of $\controller$ \cref{4_eq:Contr} is given by $\velocontroller$. Consequently, by \cref{4_thrm:veloic}, the velocity form of $\ic{\genplant}{\controller}$ is given $\ic{\velogenplant}{\velocontroller}$. Hence, $\ic{\genplant}{\controller}$ is $\qsr$-VD with $\qsr=(\perf^2I,0,-I)$ for $(\stp,\ipp)\in\BcXU$, which by \cref{4_prop:veloshifteddissip} implies that $\ic{\genplant}{\controller}$ is $\qsr$-USD for $\qsr=(\perf^2I,0,-I)$ for all $\gd\in{\tilde\gdSet}^\nnreals\cap\ltwoe$ and any $\gdEq\in\gdSetEq\cap\tilde\gdSet$. Hence, by \cref{4_lem:ls2gaindissip}, the (primal form of the) closed-loop system has a bounded \lstwo-gain of $\perf$ and, by \cref{4_thm:shiftdissipstab}, it is USAS for all $\gd\in{\tilde\gdSet}^\nnreals\cap\ltwoe$ and any $\gdEq\in\gdSetEq\cap\tilde\gdSet$.

\extver{\vspace{-1em}}{}
\proofsection{4_cor:intrealiz}\label[appendix]{4_pf:intrealiz}
For the realization of the controller in \cref{4_thrm:controlrealiz}, the input to the velocity controller is time-differentiated, which can be seen as appending a differentiator to the input of the velocity controller, see also \cref{4_fig:controlrel}. The integration filter $\weightint$, given by $\weightint(s) = \frac{s+\alpha}{s}$, is also connected to the input of the controller as depicted in \cref{4_fig:contrAndInt}. As differentiation in time can be expressed in the Laplace domain as $s$, we have that the interconnection of the weighting filter and differentiator is given by $s \cdot\weightint(s) = s \frac{s+\alpha}{s} = s+\alpha$. Hence, as $\opp = \ipk$, in the proof of \cref{4_thrm:controlrealiz}, $\dt \ipk = \ipkdot$ in \cref{4_eq:diffint} becomes $\dt\ipk + \alpha \ipk = \ipkdot$. Consequently, 
\cref{4_eq:deltaKstk} becomes
\begin{equation}
	\dt \stkdot = \lpvAk(\sch)\stkdot+\lpvBk(\sch) (\dt \ipk+\alpha \ipk),
\end{equation}
and we can write, similarly as \cref{4_eq:stkstuffpf}, that
\begin{equation}
	\dt \!\left(\stkdot-\lpvBk(\sch)\ipk\right) \!= \!\lpvAk(\sch)\stkdot - (\dt \lpvBk(\sch))\ipk+(\lpvBk(\sch)\alpha) \ipk.
\end{equation}
Defining again that $\stkA = \stkdot - \lpvBk(\sch)\ipk$, we obtain 
\begin{equation}
    \dotstkA = \lpvAk(\sch)\stkA+\left(\lpvAk(\sch)\lpvBk(\sch)+\lpvBk(\sch)\alpha I-\dif\lpvBk(\sch,\dot\sch)\right)\ipk, 
\end{equation}
Similarly, \cref{4_eq:dotskb} becomes
\begin{equation}
    \dotstkB = \lpvCk(\sch)\stkA+\left(\lpvCk(\sch)\lpvBk(\sch)+\lpvDk(\sch)\alpha I-\dif\lpvDk(\sch,\dot\sch)\right)\ipk.
\end{equation}
Then, along the same lines as in the proof of \cref{4_thrm:controlrealiz}, we obtain $\lpvAkus$, $\lpvCkus$, and $\lpvDkus$ as in \cref{4_eq:shiftcontrmat}, and $\lpvBkus$ given by \cref{4_eq:Bcontr}.
\extver{\vspace{-1em}}{}
\section{Conversion to a coarser structure}\label[appendix]{sec:app-coarse-sys}
Consider a system given by \cref{4_eq:nonlinsys} to which we serially connect (low-pass) filters, denoted by $\{\filter_i\}_{i=1}^2$, in the form of
\begin{subequations}\label{b_gpconv:lpfilt}
	\begin{align}
		\dot\stf{i}(t) &= -\Omega_i \,\stf{i}(t)+\Omega_i\,\ipf{i}(t);\\
		\opf{i}(t)&= \stf{i}(t);
	\end{align}
\end{subequations}
where $\Omega_i = \diag(\omega_{1,i},\dots,\omega_\stfSize{i})$ with $\omega_{i,j}>0$ for $j=1,\dots,\stfSize{i}$ and $\stf{i}(t)\in\reals^\stfSize{i}$.

Connecting $\filter_1$ and $\filter_2$ as given by \cref{b_gpconv:lpfilt} to \cref{4_eq:nonlinsys}, such that $\gd=\opf{1}$ and $\ipf{2}=\gp$, results in
\begin{subequations}\label{bconv_eq:nonlinSysFilt}
	\begin{align}
		\dot\stf{1}(t) &= -\Omega_1 \stf{1}(t)+\Omega_1 \hat \gd(t);\\
		\dot\stf{2} &= -\Omega_2 \stf{2}(t)+\Omega_2 \opMap(\st(t),\stf{1}(t));\\
		\dot\st(t) &= \stMap(\st(t),\stf{1}(t));\\
		\hat \gp(t) &= \stf{2}(t).
	\end{align}
\end{subequations}
The system represented by \cref{bconv_eq:nonlinSysFilt} with input $\hat \gd$, output $\hat \gp$, and state $\col(\stf{1},\stf{2},\st)$ is of the form 
\begin{subequations}\label{bconv_eq:nonlinsysState}
\begin{align}
	\dot \st(t) &= \stMap(\st(t))+\ltiB \gd(t);\\
	\op(t) &= \ltiC \st(t).
\end{align}
\end{subequations}
For \cref{bconv_eq:nonlinSysFilt} we have that $\ltiB=\begin{bmatrix}
	\Omega_1^\top & 0 & 0
\end{bmatrix}^\top $ and $\ltiC = \begin{bmatrix}
	0 & I & 0
\end{bmatrix}$, and consequently $CB=0$.

\section*{References}\vspace{-1.5em}
\bibliographystyle{IEEEtran}
\bibliography{bibFull,bibtex_db_auth}{}

\begin{IEEEbiography}[{\includegraphics[width=1in,height=1.25in,clip,keepaspectratio]{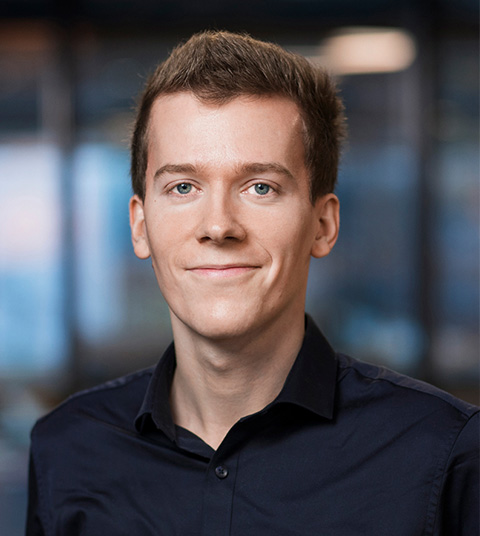}}]{Patrick J.W. Koelewijn} received his Master's and Ph.D. degrees at the Control Systems Group, Department of Electrical Engineering, Eindhoven University of Technology, both cum laude, in 2018 and 2023, respectively. During his Master's degree, he spent three months at the Institute of Control Systems at the Hamburg University of Technology (TUHH). He is currently a postdoctoral researcher at the Control Systems Group, Department of Electrical Engineering, Eindhoven University of Technology. His main research interests include the analysis and control of nonlinear and linear parameter-varying (LPV) systems, optimal and nonlinear control, and machine learning techniques.\vspace{-1em}
\end{IEEEbiography}

\begin{IEEEbiography}[{\includegraphics[width=1in,height=1.25in,clip,keepaspectratio]{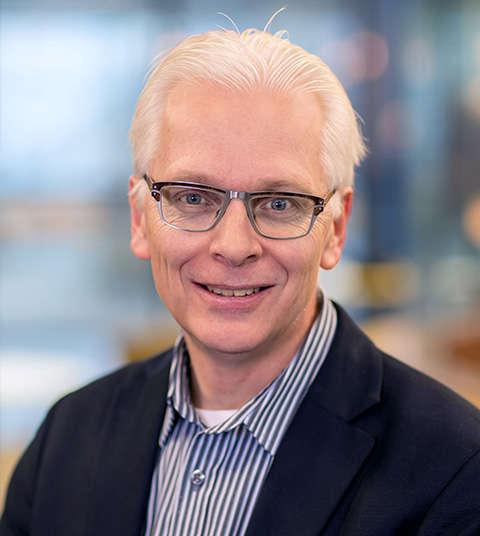}}]{Siep Weiland} received the M.Sc. (1986) and Ph.D. (1991) degrees in mathematics from the University of Groningen, The Netherlands. He was a Postdoctoral Research Associate at the Department of Electrical Engineering and Computer Engineering, Rice University, Houston, USA, from 1991 to 1992. Since 1992, he has been affiliated with Eindhoven University of Technology, Eindhoven, The Netherlands. He is a Full Professor at the same university with the Control Systems Group, Department of Electrical Engineering. His research interests are the general theory of systems and control, robust control, model approximation, modeling and control of spatial-temporal systems, identification, and model predictive control. 
\end{IEEEbiography}

\begin{IEEEbiography}[{\includegraphics[width=1in,height=1.25in,clip,keepaspectratio]{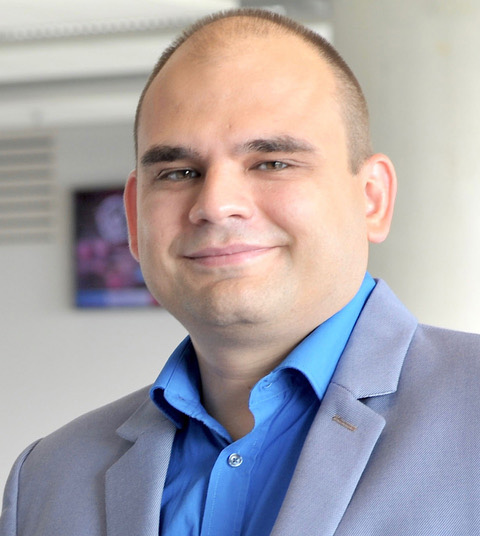}}]{Roland T\'oth} received his Ph.D. degree with cum laude distinction at the Delft Center for Systems and Control (DCSC), Delft University of Technology (TUDelft), Delft, The Netherlands in 2008.  He was a Post-Doctoral Research Fellow at TUDelft in 2009 and Berkeley in 2010. He held a position at DCSC, TUDelft in 2011-12. Currently, he is an Associate Professor at the Control Systems Group, Eindhoven University of Technology and a Senior Researcher at SZTAKI, Budapest, Hungary. His research interests are in identification and control of linear parameter-varying (LPV) and nonlinear systems, developing machine learning methods with performance and stability guarantees for modelling and control, model predictive control and behavioral system theory.\vspace{-1em}
\end{IEEEbiography}

\end{document}